\documentclass[conference]{IEEEtran}
\usepackage{amsmath,algorithm}
\usepackage{tablefootnote}
\usepackage{pifont}
\newcommand{\cmark}{\ding{51}}%
\newcommand{\xmark}{\ding{55}}%
\usepackage{algpseudocode}
\usepackage{graphics}
\usepackage{graphicx}
\usepackage{amsmath}
\usepackage{tikz}
\usepackage{soul}
\usepackage{xspace}
\usepackage{color,balance}
\usepackage{multirow}
\usepackage{hyperref}
\let\labelindent\relax
\usepackage{enumitem}
\usepackage{amssymb}
\usepackage[utf8]{inputenc}
\usepackage{cryptocode}
\usepackage{booktabs,times}
\usepackage{float}
\usepackage{placeins}
\usepackage{caption,paralist}
\usepackage{cleveref}
\usepackage{multirow}
\usepackage{enumitem,kantlipsum}
\usepackage[belowskip=-8pt,aboveskip=0pt]{caption}

\newcommand{\sys}{\mbox{$T^3$}\xspace}

\renewcommand\paragraph[1]{\vspace{0.3mm}\noindent \textbf{#1\ }}

\newtheorem{theorem}{Theorem}

\newtheorem{claim}[theorem]{Claim}

\newtheorem{definition}{Definition}

\newcommand{\PAccess}[1]{\mathsf{Access}_\Sigma(#1)}
\newcommand{\ORequest}[1]{\stackrel{\rightarrow}{#1}}

\newenvironment{todo-text}{\color{Red}}{\hfill}

\newcommand{\SystemName}{$T^3$\xspace}

\newcommand*\rcircled[1]{\tikz[baseline=(char.base)]{
    \node[fill=black,text=white, shape=circle,draw=black,inner sep=.6pt] (char) {#1};}}
\newcommand*\bcircled[1]{\tikz[baseline=(char.base)]{
    \node[shape=circle,draw,inner sep=.6pt] (char) {#1};}}

\newcommand{\managingTEE}{\textit{managing} TEE\xspace}
\newcommand{\readTEE}{\textit{reading} TEE\xspace}
\newcommand{\updateTEE}{\textit{writing} TEE\xspace}

\newcommand{\managingEnclave}{\textit{managing} enclave\xspace}
\newcommand{\readEnclave}{\textit{reading} enclave\xspace}
\newcommand{\updateEnclave}{\textit{writing} enclave\xspace}

\newcommand{\readTree}{\textit{read-once} ORAM tree\xspace}
\newcommand{\updateTree}{\textit{original} ORAM tree\xspace}

\newcommand{\readaccessname}{\textit{read-once} ORAM access\xspace}
\newcommand{\readname}{\textit{\textit{read-once}}\xspace}

\newcommand{\mEnclave}{$\mathcal E_{m}$\xspace}
\newcommand{\rEnclave}{$\mathcal E_{r}$\xspace}
\newcommand{\uEnclave}{$\mathcal E_{w}$\xspace}

\newcommand{\evict}{$\mathsf{Eviction}$}

\newcommand{\pathORAM}{\ensuremath{\textsc{Path\text{-}ORAM}}\xspace}
\newcommand{\circuitORAM}{\ensuremath{\textsc{Circuit\text{-}ORAM}}\xspace}

\newcommand{\bci}{\textit{block creation interval}\xspace}
\newcommand{\cmov}{$\mathsf{cmov}$\xspace}

\title{
A Tale of Two Trees:
One Writes, and Other Reads\\
{\large Optimized Oblivious Accesses to Large-Scale Blockchains}}
\begin{document}
\author{\IEEEauthorblockN{Duc V. Le}
\IEEEauthorblockA{Purdue University}
\and
\IEEEauthorblockN{Lizzy Tengana Hurtado}
\IEEEauthorblockA{National University of Colombia}
\\
\IEEEauthorblockN{Byoungyoung Lee}
\IEEEauthorblockA{Seoul National University}
\and
\IEEEauthorblockN{Adil Ahmad}
\IEEEauthorblockA{Purdue University}
\\
\IEEEauthorblockN{Aniket Kate}
\IEEEauthorblockA{Purdue University}
\and
\IEEEauthorblockN{Mohsen Minaei}
\IEEEauthorblockA{Purdue University}
}

\maketitle        

\begin{abstract}
The Bitcoin network has offered a new way of securely performing financial transactions over the insecure network.
Nevertheless, this ability comes with the cost of storing a large (distributed) ledger, 
which has become unsuitable for personal devices of any kind. 
Although the simplified payment verification (SPV) clients can address this storage issue, a Bitcoin SPV client has to rely on other Bitcoin nodes to obtain its transaction history and the current approaches offer no privacy guarantees to the SPV clients.

This work presents \SystemName, a trusted hardware-secured Bitcoin full client that supports efficient and scalable oblivious search/update for Bitcoin SPV clients without sacrificing the privacy of the clients. 
In this design, we leverage the trusted execution and attestation capabilities of a trusted execution environment (TEE) and the ability to hide access patterns of oblivious random access memory (ORAM)
to protect SPV clients' requests from a potentially malicious server. 
The key novelty of \SystemName lies in the optimizations introduced to conventional oblivious random access memory (ORAM),
tailored for expected SPV client usages.
In particular, by making a natural assumption about the access patterns of SPV clients, 
we are able to propose a two-tree ORAM construction that overcomes the concurrency limitation associated with traditional ORAMs. 
We have implemented and tested our system using the current Bitcoin Unspent Transaction Output database. 
Our experiment shows that the system is highly efficient in practice while providing strong privacy and security guarantees to Bitcoin SPV clients. 
\end{abstract}
\section{Introduction}
\label{sec:introduction}

Over the last few years, we have seen a great interest in public blockchain in the community. 
The Bitcoin blockchain offered a way to provide security and privacy for financial transactions. 
However, due to the huge adoption by the community, the size of the Bitcoin blockchain has become too large for small and resource-constrained devices such as personal laptops or mobile phones, raising not only performance but also privacy concerns in the community.
As of October 2018, the size of the unindexed Bitcoin blockchain is 230 GB.

To this end, Bitcoin's simplified payment verification (SPV) client has become a widely-adopted solution to resolve a storage problem for constrained devices. 
Nakamoto~\cite{Nakamoto_bitcoin:a} sketched the idea of SPV clients in the Bitcoin whitepaper, 
and in the Bitcoin improvement proposal 37 (BIP37)~\cite{BIP37}, Mike Hearn combines Nakamoto's idea with the use of Bloom filters to standardize the design of Bitcoin SPV clients. 
This design has become de facto standard and been used by other SPV clients such as BitcoinJ~\cite{bitcoinj-cite} and Electrum~\cite{electrum-cite}.

The core of SPV clients is in only downloading and then verifying part of the blockchain that is relevant to the SPV client itself.
In particular, the SPV client loads its addresses into a Bloom filter and sends the filter to a Bitcoin full client, and 
The Bitcoin full client will use that filter to identify if a block contains transactions that are relevant to the SPV client, 
and once it finds the block, it will send a modified block that only contains relevant transactions along with Merkle proofs for those transactions. 

However, the current SPV solution relied on Bloom filters raises security and privacy concerns to the SPV clients when communicates with potentially malicious nodes.
In particular, Gervais et al.~\cite{Gervais:2014:SPV-privacy} show that it is possible for a malicious node to learn several addresses of the client from the Bloom filter with high probability.
Moreover, if the adversarial node can collect two filters issued by the same client, then a considerable number of addresses owned by the client will be leaked. 

To provide a strong privacy guarantee for SPV clients, one needs a solution that can hide wallets/addresses queried by the SPV clients. 
While such a system can be built using private information retrieval (PIR) primitive, the existing cryptographic PIR solutions~\cite{JaschkeGAS17} are not been practical to scale to handle millions of Bitcoin users.
On the other hand, to gain more efficiency, one can use ORAM and trusted execution environment to propose generic PIR systems~\cite{thang-hoang:posup-popets,oblidb,SasyGF18-zero-trace}. 
However, as it becomes apparent in the later in this paper, naively combining ORAM scheme as it is with TEE makes 
the practicality of those generic systems questionable when used in a large network like Bitcoin due to the lack of concurrency in ORAM as well as the limitation of TEE with restricted memory.


\noindent\textbf{Our Contribution.} 
This work aims not only to design a system that provides SPV clients with privacy-preserving access to the Bitcoin blockchain data but also to consider other practical aspects on how to scale such a system to handle client requests in a large-scale.
Our contributions can be summarized as follows:

Firstly, we present a novel design for a system that can handle up to thousands of requests per minute from Bitcoin SPV clients based on a 
\textit{restricted access} Oblivious Random Access Memory (ORAM) and the trusted execution capabilities of TEE. 
In particular, one of the main contributions of our design is the optimization access in the prominent tree-based ORAM schemes
that allow those ORAM schemes to support concurrent accesses which is essential for handling SPV clients' requests. 
In this design, the access privacy guarantee is still maintained because of our natural assumption that the rational Bitcoin SPV clients should only query for their particular transaction {\em once} before the arrival of a new Bitcoin block. 
Nevertheless, we later show that even when the SPV clients are irrational then the privacy for such clients is only compromised for a short period of time.  
The security guarantee of \sys also relies on the trusted execution capabilities of TEE that allows SPV clients to perform ORAM operations securely and remotely.
Our generic design works with other blockchains, any tree-based ORAM schemes~\cite{Elaine-rORAM,Stefanov:2013,wang-circuit-oram-2015}, and any TEE with attestation capability.



Secondly, we implemented a prototype of \sys and evaluated its performance to demonstrate the practicality of our approach. 
More specifically, we extracted the unspent transaction outputs set of Bitcoin in October 2018 and used it to measure the performance of the system when handling clients' requests. 
The implementation of \sys also adopts standard techniques (i.e., oblivious operations using $\mathsf{cmov}$~\cite{SasyGF18-zero-trace,ndss-AhmadKSL18,racoon}) to be secure against known side-channel attacks~\cite{shadow-branch-lee-usenix17,hid-sgx-sidechannel-usenix17,Xu15ControlledChannel}.
Moreover, the use of recursive ORAM constructions in \sys makes the system much more suitable for TEE with restricted trusted memory like Intel SGX. 
We then show that the running time of the ORAM read access decreases linearly with the number of the threads used (e.g, up to $8\times$ performance gained with $4$ threads). 
%



Finally, we conclude that putting natural restrictions on the access patterns on oblivious memory can lead to significant performance improvement and better ORAM design.  While the applicability of \SystemName in cryptocurrencies beyond Bitcoin is apparent, we believe our work will also motivate further research on oblivious memory with restricted access patterns.

\paragraph{Concurrent Work.}
The soon-to-be published BITE system~\cite{bite-spv-sgx} also employs the Oblivious Database construction for SPV client privacy.
The main idea of the BITE construction is to combine the use of non-recursive \pathORAM~\cite{Stefanov:2013} construction and TEE (such as Intel SGX) to propose a generic system that offers SPV client with oblivious access to the database. 
However, BITE did not address several shortcomings of using \pathORAM as it is and TEE with restricted memory in practice. 
In particular, the BITE design did not consider use recursive ORAM constructions to reduce the trusted memory usage; therefore, the efficiency of the system will be degraded once the size of the database gets too large.  
Moreover, due to the inherent lack of concurrency in tree-based ORAM such as \pathORAM, naively using Path-ORAM makes BITE unsuitable for handling thousands of Bitcoin client's requests per minute as well as thousands of updates every fixed period of times (e.g., 10 minutes for Bitcoin).
In this work,  we investigate the use of both recursive \pathORAM and recursive \circuitORAM to understand the actual performance and the actual storage overhead put on the server. Importantly, we propose a two-tree ORAM design to further enhance the performance of standard ORAM accesses as well as to allow concurrent requests from the SPV client. 
\begin{figure*}[h!]
    \centering
    \includegraphics[scale=.65]{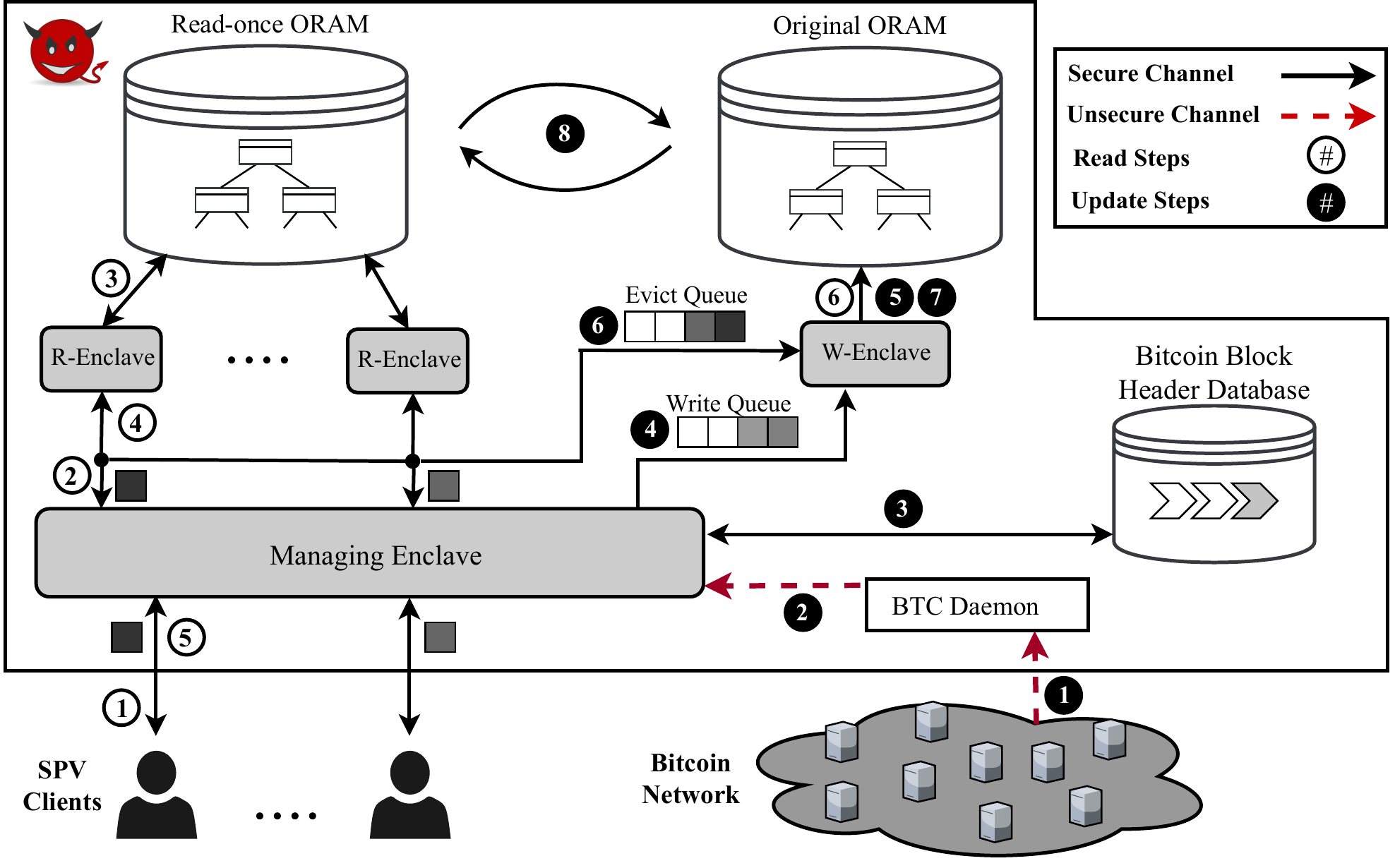}
    \caption{An overview of \sys workflow. The encrypted ORAM databases are stored in server untrusted memory region. Steps \protect\bcircled{1}-\protect\bcircled{6} describe the flow of the request sent from SPV clients. Steps \protect\rcircled{1}-\protect\rcircled{8} describe the flow of the update procedure when \sys receives new Bitcoin block} 
	\label{sub:system_overview}
\end{figure*}
\section{Design Goals and Solution Overview} 
\label{sec:Overview}
In this section, we define the system components, outline our security goals, and give an overview of how our system works.
\subsection{System Components} 
\label{sub:system_components}
There are three key components of this system: the Bitcoin network, a client, and a untrusted server.
	The \textbf{Bitcoin network} is a set of nodes that maintains the Bitcoin blockchain, and the network validates and relays the new Bitcoin block produced by miners. 
%
	A \textbf{client} is a Bitcoin simplified payment verification node that remotely connects to the secure TEE on the untrusted server to perform oblivious searches on the unspent transaction output (UTXO) set. The client is also able to connect to the Bitcoin Network to obtain other network metadata such as the latest Bitcoin block header. 
	A \textbf{server} is the untrusted entity made up of two components: an untrusted server and several trusted TEEs (i.e., the \textit{managing}, \textit{reading}, and \textit{writing} TEEs). Moreover, the untrusted server stores three encrypted databases which are the \readTree, the \updateTree, and the Bitcoin header chain. The untrusted server hosts a potentially malicious bitcoin client (e.g., $\mathsf{bitcoind}$) that handles the communication with the Bitcoin Network. 

\subsection{Design Goals} 
\label{sub:goal}
The goal of our system is to leverage the trusted execution capabilities of Trusted Execution Environment (TEE) with attestation to design a public Bitcoin full node that supports oblivious search and update on the current Bitcoin unspent transaction output database. 
%
Our system aims to provide data confidentiality and privacy to Bitcoin SPV clients in a large scale by using standard encryption and Oblivious RAM techniques on the current set of unspent transaction outputs. 
The main goals that \sys tries to achieve are:
\begin{compactenum}
	\item \textbf{Privacy.}     \sys aims to provide privacy and confidentiality to SPV clients' requests. 
	                    In particular, 
	                    the system allows SPV clients to obliviously search its relevant transactions without revealing their addresses to potentially malicious providers by using TEE to encrypt the data and using ORAM schemes to eliminate known side channel leakages~\cite{racoon,ndss-AhmadKSL18,thang-hoang:posup-popets,SasyGF18-zero-trace}.  
	\item \textbf{Validity.}    In our design, the SPV client should be able to obtain valid information based on the provided addresses, 
	                    and a malicious adversary should not able to tamper the Blockchain data with invalid transaction outputs. 
	\item \textbf{Completeness.} The system should provide clients with access to most of its relevant transactions in order to determine balance or to obtain essential information to form new transactions.
	\item \textbf{Efficiency.}  The system should be practical to deploy. 
						More specifically, the system should be efficient enough to handle different concurrent SPV clients' requests without compromising the privacy of the clients. 
\end{compactenum}


\subsection{Solution Overview} 
	The idea of using ORAM schemes and trusted execution environment to construct database systems that support oblivious accesses has been investigated by the research community~\cite{thang-hoang:posup-popets,oblidb,SasyGF18-zero-trace}.
	However, the efficiency and scalability of those systems are hampered by the lack of concurrency of traditional ORAM schemes~\cite{Stefanov:2013,wang-circuit-oram-2015}.

	In this work, we design \sys to overcome the limitations of efficiency and
	concurrency plaguing existing systems. Our design is motivated by
	the following observations. 
	The first observation is that each ORAM access in a standard tree-based ORAM settings is a combination of two operations: a \textit{read-path} operation and an \textit{eviction} operation.
	By separating the effects two operations into two different trees: a \readTree and an \updateTree, one can use read-path operation on the \readTree to handle clients' requests simultaneously while performing a non-blocking eviction operation on the \updateTree sequentially.
	This design is also independently investigated by ConcurORAM~\cite{concurORAM-chakraborti}; however, their design is not suited
	for TEEs with limited trusted memory (such as Intel SGX). We elaborate on this
	in the coming sections.
	
	The second observation is that the access privacy guarantee of this approach relies the characteristic of the Bitcoin blockchain. 
	In particular, the Bitcoin network generates new Bitcoin block on average of 10 minutes, and if we require \sys to periodically synchronize these the two trees, then the privacy of clients' queries are preserved. 
	Moreover, if we assume that upon receiving transactions belonged to its addresses, the rational client should not query same transactions {\em again} until the next block arrives, the proposed approach on the separation of \textit{read-path} and \textit{eviction} procedure not only does not affect the privacy guarantees of ORAM access but also allows \sys to handle much more clients' requests. 
	More importantly, we also argue that even when the SPV clients are irrational by submitting requests for the same transaction more than one, 
	the privacy of those clients is only compromised for a short period of time (i.e., 10 minutes for the Bitcoin network) because \sys will always synchronize the old instance of the \readTree with the more updated instance of the \updateTree.
	With the intuition of \sys described above, we outline the workflow
	of our design:

\noindent	\textbf{Server Initialization}~\rcircled{1}-\rcircled{8}: 
	Initially, the managing TEE will initialize a \updateTEE that creates an empty ORAM tree. 
	For each of Bitcoin block obtained from the network, the managing TEE verifies the proof of work of the block before passing relevant update data to the \updateTEE in order to populate the ORAM tree.
	With the current size of the Bitcoin blockchain, this operation might take several hours. 
	However, once the TEEs catch up with the current state of the Bitcoin blockchain, we expect that the TEE only has to perform a batch of update accesses on the ORAM tree every 10 minute.
	When the initialization is completed, the \managingTEE creates two copies of the ORAM tree which are the \readTree and the \updateTree. 

\noindent\textbf{Oblivious \readname Protocol}~\bcircled{1}-\bcircled{6}: In order to obtain its unspent outputs, the client first performs the remote attestation to the \managingTEE. 
	The remote attestation mechanism allows the client to verify the correctness of program execution inside the TEE. 
	More importantly, after a successful attestation, the client can use standard key exchange mechanism (i.e. Diffie-Hellman's key exchange) to share a secret session key with the TEE in order to establish a secure connection with the \managingTEE. Upon receiving client's connection requests, the \managingTEE creates a \readTEE with its own copies of the ORAM position map and the ORAM stash to handle client subsequent requests. 
	Next, after having a secure channel, the client will send his Bitcoin addresses along with the proof of ownership of those addresses to the TEE (e.g., the knowledge of the public key along with a signature to a random nonce or the preimage of the public key hash). 
	The \readTEE will use a mapping function to map Bitcoin addresses into the ORAM block identification number and performs \readaccessname on the ORAM tree. 
	In particular, those \readaccessname do {\em not} involve the eviction procedure which requires re-encrypting and remapping the ORAM block. 
	The eviction procedure will be performed on the \updateTree~by the \updateTEE.
	
    
\noindent \textbf{Oblivious Write Protocol}~\rcircled{1}-\rcircled{8}: 
	The \sys requires to update the ORAM tree via batch of write accesses every 10 minutes on average. 
	In particular, 
	\sys will rely on a standard Bitcoin client 
	to handle the communication with the Bitcoin network to obtain blockchain data\footnote{This ability can be easily included in the future implentation of \sys.}.
	Thus, \sys needs to verify the block relayed by a potentially malicious Bitcoin client before updating the ORAM tree. 
	More specifically, in the design,  \sys stores a separate Bitcoin header chain to verify the proof of work and the validity of all transactions inside a Bitcoin block.
 	After the verification, the \managingTEE forms a batch of ORAM updates and delegates those updates to the \updateTEE. 
	Once those updates are finished, the \managingTEE will queue up read requests from SPV clients in order to allow the \updateTEE to finish the eviction requests from the $\textit{read}$ TEEs during the updating interval. 
	As soon as the \updateTEE finishes performing those eviction requests, 
	the \managingTEE updates the position map and \textit{stash}, and makes the ORAM tree used by the \updateTEE become the new ORAM tree used by \readTEE.
	At this point, the \readTEE can use the new tree instance to respond to clients' requests while the \updateTEE performs the eviction procedure on another copy of the same ORAM tree.

\section{Preliminaries and Threat Model} 
\label{sec:preliminaries}
\subsection{Trusted Execution Environment} 
\label{sub:intel_sgx}
The design of \SystemName relies on a trusted execution environment (TEE) to prove the correctness of the computations. 
In particular, TEE is a trusted hardware that provides both confidentiality and integrity of computations as well as offer an authentication mechanism, known as \textit{attestation}, for the client to verify computation correctness.
In this work, we chose Intel SGX~\cite{sgx-explained} to be the building block of our system.
However, with minor modifications, the design of our system can be extended to any TEE with \textit{attestation} capabilities such as Keystone-enclave~\cite{keystone-project} and Sanctum~\cite{constant-sanctum-usenix} as other trusted execution environments might not have the same strengths/weaknesses as Intel SGX.

Intel SGX is a set of hardware instructions introduced with the 6th Generation Intel Core processors. 
We use Intel SGX as a TEE for the execution of an ORAM controller on the untrusted server. 
The relevant elements of SGX are as follows.
	\textbf{Enclave} is the trusted execution unit that is located in a dedicated portion of the physical RAM called the enclave page cache (EPC). The SGX processor makes sure that all other software components on the system cannot access the enclave memory. 
%
	Intel SGX supports both {\bf local and remote attestation} mechanisms to allow remote parties or local enclaves to authenticate and verify if the program is correctly executed within an SGX context.
	More importantly, attestation protocols provide the authentication required
	for a key exchange protocol \cite{sgx-explained}, i.e.,
	after a successful attestation, the concerned parties can agree on a shared session key using Diffie-Hellman Key Exchange~\cite{DH-keyexchange} and create a secure channel.

\noindent\textbf{Limitations.}
Intel SGX comes with various limitations which have been uncovered by the
academic community over the past few years.
In particular, some of the limitations are:
\begin{compactitem}
	\item \textbf{Side Channel Attacks:} While Intel SGX provides security guarantees against direct memory attacks, it does not provide systematic protection mechanisms against side channel attacks
	such as page table-based~\cite{Xu15ControlledChannel, hid-sgx-sidechannel-usenix17}, cache-based~\cite{cache-based-attack}, and branch-prediction-based~\cite{shadow-branch-lee-usenix17}.
	Through page table and cache attacks, a privileged attacker can
	observe cache-line-granular (i.e., 64B) memory access patterns from the enclave program. On the other hand, the branch-prediction attack
	can potentially leak all the control-flow taken by the enclave program.

	\item \textbf{Enclave Page Cache Limit:} The size of the Enclave Page Cache (EPC) is limited to around $96$MB~\cite{Arnautov-epc}. Although Intel SGX alleviates this limitation by supporting page-swapping between trusted memory region and untrusted memory region, this operation is expensive due to encryption and integrity verification \cite{sgx-explained,Arnautov-epc}. 

	\item \textbf{System Calls:} Intel SGX programs are restricted to ring-3 privileges and therefore rely on the untrusted OS for ring-0
	operations such as file and network I/O. 
	There are various previous works which try to solve this problem using library OSes~\cite{lib-oses} 
    and/or other techniques~\cite{thang-hoang:posup-popets}. 
\end{compactitem}

\noindent
\textbf{Oblivious Operations inside the Enclave.}
	Several techniques~\cite{racoon,thang-hoang:posup-popets,SasyGF18-zero-trace,oblivious-olga-ohrimenko} have been introduced to mitigate side-channel attacks on the SGX. 
	In this work, we built our system based on the implementations of both $\mathsf{Zerotrace}$~\cite{SasyGF18-zero-trace} and $\mathsf{Obliviate}$~\cite{ndss-AhmadKSL18}.
	Therefore, our system inherited standard secure operations from both of these libraries. 
    In particular, their implementations use an oblivious access wrapper
    by using the x86 instruction \cmov as introduced by Raccoon~\cite{racoon}.
    Using \cmov, the wrapper accesses every single byte of a memory object
    while reading or writing only the required bytes in memory.
    From the perspective of an attacker (which can only observe
    access-patterns), this is the same as reading or modifying
    every byte in memory.
	We refer readers to \cite{SasyGF18-zero-trace,ndss-AhmadKSL18, racoon} 
	for detailed description of these oblivious operations.

\subsection{Oblivious Random Access Memory} 
\label{sub:oblivious_ram}
Oblivious Random Access Memory (ORAM) was first introduced by Goldreich et al~\cite{Goldreich:1987} for software protection against piracy. The core of ORAM is to hide the access patterns resulted by reading and writing accesses on encrypted data. The security of ORAM can be described as follows.

\begin{definition}~\cite{Stefanov:2013} 
	Let
	$\stackrel{\rightarrow}{y}=(\mathsf{op_i,bid_i,data_i})_{i\in [n]}$
	denote a sequence of accesses 
	where $\mathsf{op_i}\in \{read,write\}$, 
	$\mathsf{bid_i}$ is the identifier, 
	and $\mathsf{data_i}$ denotes the data being written. 
	For an ORAM scheme $\Sigma$, let $\mathsf{Access}_{\Sigma}(\stackrel{\rightarrow}{y})$ denote a sequence of physical accesses pattern on encrypted data produced by $\stackrel{\rightarrow}{y}$.
	We say:
	\begin{inparaenum}[(a)]
	  	\item The scheme $\Sigma$ is secure if for any two sequences of accesses $\stackrel{\rightarrow}{x}$ and $\stackrel{\rightarrow}{y}$ of the same length, $\PAccess{\ORequest{x}}$ and $\PAccess{\ORequest{y}}$ are computationally indistinguishable.
	  	\item The scheme $\Sigma$ is correct if it returns on input $\ORequest{y}$ data that is consistent with $\ORequest{y}$ with probability $\geq 1 - \mathsf{negl}(\lvert\ORequest{y}\rvert)$ i.e {negligible in $\lvert\ORequest{y}\rvert$}
  	\end{inparaenum}
\end{definition}

\noindent
\textbf{Tree-based ORAM schemes. } One strategy of designing an ORAM scheme is to follow the tree paradigm proposed by Shi et al.~\cite{Elaine-rORAM} and Stefanov et al.~\cite{Stefanov:2013}. 
In tree based ORAM, the client encrypts their database into $N$ different encrypted data blocks and obliviously stores those data blocks in a binary tree of height $\lceil \log_2(N) \rceil$. 
Each node in the tree is called a \textit{bucket}, and each \textit{bucket} can contain up to $Z$ blocks. 
The client also maintains a \textit{position map}, to indicate which path a data block resides on. 
Finally, the client needs to have a \textit{stash} to store a path retrieved from the server.

We follow the same generalization of a tree-based ORAM access described in~\cite{thang-hoang:posup-popets}. 
Each access in both ORAM schemes requires two operations: a $\mathsf{ReadPath}$ operation and an $\mathsf{Evict}$ operation. 
Intuitively, $\mathsf{ReadPath}$ takes as input the ORAM block identifier, $\mathsf{bid}$, accesses the position map, and retrieves the path that block $\mathsf{bid}$ resides onto the stash, $S$. 
After performing ORAM access (i.e. $\mathsf{read/write}$) on the identified block, the block is assigned to a different path and pushed back to the tree via the $\mathsf{Evict}$ operation.
In general, the $\mathsf{Evict}$ operation takes a stash and the assigned path as input, writes back blocks from stash to the assigned path, and update the position map. 
Figure~\ref{fig:oram-access} gives an overview of how tree-based ORAM access works.

\begin{figure}[t]
	\centering
	\begin{minipage}{\linewidth}
	\begin{algorithm}[H]
	  	\caption{$ORAM.\mathsf{Access}(\mathsf{op,bid,data^*})$}
	  	\begin{algorithmic}[1]
			\State $S \leftarrow \mathsf{ReadPath}(\mathsf{bid})$ 
			{\color{blue}{//scan the whole stash}}
			\State $\mathsf{data}\leftarrow$ block $\mathsf{bid}$ from $S$ 
			\If{$\mathsf{op} = write$}
			\State $S$ $\leftarrow$ $(S$$ - \mathsf{\{(bid,data)\})\cup \{(bid,data^*)\}}$
			\EndIf
			\State $p' \stackrel{\$}{\leftarrow} \{0, \dots, N-1\}$        {\color{blue}{//the random eviction path is selected}}
			\State $S\leftarrow \mathsf{Evict}(S,p')$ 
			\State \Return $\mathsf{data}$
		\end{algorithmic}	
	\end{algorithm}
	\end{minipage}
	\caption{a standard tree-based ORAM \textit{read/write} access}
	\label{fig:oram-access}
\end{figure}
\noindent
\textbf{\pathORAM/\circuitORAM scheme.}
In this work, we consider two popular tree-based constructions of ORAM which are $\pathORAM$~\cite{Stefanov:2013} and $\circuitORAM$~\cite{wang-circuit-oram-2015}. 
While $\pathORAM$ offers simple $\mathsf{ReadPath}$ and $\mathsf{Evict}$ operations, $\circuitORAM$ offers a smaller circuit complexity for the $\mathsf{Evict}$ procedure. 
Thus, $\circuitORAM$ is more efficient when implemented with Intel SGX. 
As noted in \cite{SasyGF18-zero-trace,thang-hoang:posup-popets,wang-circuit-oram-2015}, $\circuitORAM$ can operate with $Z=2$ compared to $Z=4$ as in $\pathORAM$; therefore, the server storage overhead is significantly reduced. Moreover, the size of \textit{stash} in $\circuitORAM$ is smaller compared to the size of \textit{stash} in $\pathORAM$; this allows a more efficient performance when scanning the stash as one needs to scan the whole path and stash to avoid side-channel leakage.
\noindent
\noindent
\textbf{Recursive ORAM.} 
In a non-recursive tree-based ORAM setting, the client has to store a position map of the size $O(N)$ bits. 
This approach, however, is not suitable for a resource-constrained client. 
Stefanov et. al~\cite{Stefanov:2013} presented a technique that reduces the size of the position map to $O(1)$. 
The main idea of those constructions is to store a position map as another ORAM tree in the server side, and the client only keeps the position map of the new ORAM.
The client keeps compressing the position map into another ORAM tree until the size of the position map is small enough to be saved on the client's storage. 
One main drawback of those constructions is the increased cost in the communication between a client and the server. 
Fortunately, in our setting, this cost can be safely ignored because the communication between client and server becomes the I/O access between TEE and the random access memory.
Thus, it is more reasonable to use recursive constructions because it reduces the memory stored in the trusted region (e.g., Processor Reserved Memory).

\subsection{Blockchain}
\label{sub:btc}
The Bitcoin blockchain is a distributed data structure maintained by a group of nodes. 
In this work, to simplify the structure of the Bitcoin blockchain, we denote the network as a single party that maintains a growing database of Bitcoin blocks.
On average of $10$ minutes, the network outputs a Bitcoin block which is a combination of Bitcoin transactions and a block header.
Each block header contains relevant information about the Bitcoin block such as Merkle root, nonce, network difficulty.
The Merkle root can be used to verify the membership of Bitcoin transactions, 
and the nonce and difficulty are used to check the proof of work.
Each Bitcoin transaction contains a set of inputs and outputs where transaction inputs are unused outputs of previous transactions.
	
\begin{asparaitem}
	\item\textbf{Unspent Transaction Output Database.}	
		In the Bitcoin network, the balance of a Bitcoin address is determined by values of those outputs that have not been used in other transactions.
		These outputs are called Unspent Transaction Outputs (UTXO).
		Moreover, in the implementation of common Bitcoin nodes such as Bitcoin core \cite{bitcoin-core}, 
		Bitcoin nodes maintain a separate database that keeps track of all unspent transaction outputs and other metadata of the Bitcoin blockchain. 
		This database is known as the UTXO set.
		Intuitively, a client with the knowledge of the secret key and the commitment value can query the UTXO set directly to obtain essential information such as transaction hash, position, and value to form new valid transactions.
		Therefore, in this work, we realize that if a full node can securely update and maintain the integrity of the UTXO set via while provides SPV clients with oblivious accesses to the UTXO set, the privacy of the SPV client is preserved.

	\item \textbf{Bitcoin transaction types.}
	In the Bitcoin blockchain, transactions are classified based on the structure of the input and output scripts. 
	In particular, there are five types of standard script templates which are \textit{Pay-to-Pubkey} (P2PK), \textit{Pay-to-PubkeyHash} (P2PKH), \textit{Pay-to-ScriptHash} (P2SH), \textit{Multisig}, and \textit{Nulldata}. 
	Intuitively, scripting in Bitcoin provides a way to prove the ownership of the coins. 
	In particular, a challenge script (\textit{scriptPubkey}) is included as a part of the transaction output to specify the condition for its redemption, and a response script (\textit{scriptSig}) is part of the transaction input to reveal the condition needed to redeem the bitcoins from other output. 

	In this work, we only consider two types of transaction: \textit{Pay-to-PubkeyHash} (P2PKH) transaction and $\textit{Pay-to-ScriptHash}$ (P2SH) transaction.
	According to \cite{analysis-of-utxo,mohsen-r3c3}, these two types of transaction make up of $97\text{-}99\%$ of the UTXO set. 
	Also, one can assume that the \textit{Pay-to-Pubkey-Hash} transaction is one variant of the \textit{Pay-to-Script-Hash} transaction 
	because both transaction types require the spender's knowledge of the preimage of the hash digest before being able to spend those outputs. 
	For simplicity, from this point on, we assume that the only information needed to obtain the unspent output is the public key hash, $pkh$. 
	Moreover, all other transaction types such as \textit{Multisig} and P2PK can be easily supported in the future.
	
	\item \textbf{Block creation interval. } The block creation time in Bitcoin is the time that the network takes to generate a new block, and block creation time is specified to be 10 minutes on average by the network. 
	We call the waiting period between the most recent block and a new block, \textit{block creation interval}. 
	In this work, we discretize time as \textit{block creation intervals}. 
\end{asparaitem}


\subsection{Threat Model}
\label{sub:Threat_Model}
We assume that SPV \textit{clients} are honest and rational which means that before during the \textit{block creation interval},
a SPV \textit{client} should not request the server for transaction outputs of a same public key hash more than once. 
The underlying remote attestation service provided by Intel is secure and trusted. 
The local attestation between enclaves is secure. 
The server and its programs are assumed to be untrusted except for programs running within an enclave. 

We assume that the adversary who controls the operating system can read/inject/modify encrypted messages sent by enclaves. 
The adversary also can observe memory access patterns of both trusted and untrusted memory. Also, the computation power of the adversary is assumed to be limited. 
In particular, during \textit{block creation interval}, the adversary should not have enough computation power to forge a new Bitcoin block that satisfies the current Bitcoin network difficulty. As the time of writing, the network difficulty~\cite{bitcoin-difficulty} is around $6\times 10^9$; 
therefore, the expected number of hashes to mine a Bitcoin block is roughly $2^{72}$.  

The server's attacks on availability are out of scope.
More specifically, denial of service (DoS) attacks by system admin and untrusted operating system are out of the scope.
Otherwise, such adversary can prevent the enclaves from receiving new bitcoin block by shutting down the communication channel between the enclave and the Bitcoin network as the enclave has to rely on the untrusted OS to perform system calls such as file and network I/O. 
\section{Proposed System} 
\label{sec:new-protocol}

In this section, we first describe how \SystemName stores the UTXO set by exploring different mappings between the unspent transaction outputs and the ORAM blocks.
We see that naive mapping may not be secure for blockchain applications as it may lead to denial of services attacks. 
Next, we demonstrate how Intel SGX can be considered as a trusted execution unit to access ORAM and perform read/write operations in an oblivious manner. 
Finally, we will describe how the system handles clients' requests during a write operation.

\subsection{Storage Structure of the UTXO set}
\label{subsec: Oblivious Storage of the UTXO set}
In the first step, we show how to map the public key hash to ORAM block identification and then describe the storage requirements in \SystemName. 

\subsubsection{Bitcoin unspent transaction output mapping}
\label{subsubsec: btcintoORAM}

In the design of \sys, we assume that the SPV clients only know his/her addresses (i.e., the public key hashes); 
therefore, to return the outputs belonging to the client's address, the enclave needs to know the mapping between the address and the ORAM block identification. 


In this work, we propose two simple mappings to store unspent outputs in the ORAM tree. 
More specifically, both approaches use standard cryptographic hash functions along with a secret key generated by the enclave.
The first approach is to map a single Bitcoin address into a single ORAM block, 
and the second approach is to map a Bitcoin address into multiple ORAM blocks. 
We will later explain the trade-off between these two approaches. 
Intuitively, the first approach is more efficient in terms of performance and can be more expensive in terms of storage overhead. 
The second approach gives some flexibility in terms of storage overhead; however, to offer strong privacy to every address, this approach can be more expensive in terms of performance because it may incur more ORAM calls. 

\noindent
\paragraph{Single address into Single ORAM block.}
\label{par:key-hash-with-oram}
In this design, during the initialization, we require the program inside the enclave to use a keyed hash function to map the public key hash to ORAM block identification. 
The secret key of the hash function is generated and known only by the enclave.
In other words, the mapping between a Bitcoin address to an ORAM block identification is known only to the SGX.
%
We define the mapping as follow:
\begin{compactitem}
	\item $\mathsf{bid}\leftarrow \mathsf{OBlockMap}({pkh}, k_b)$: the function takes as input a $20$-bytes hash digest ${pkh}$ and a secret key $k_b$, it outputs the block identification number $\mathsf{bid}\in \{0,\dots,N-1\}$. 
\end{compactitem}

The key-hashing approach offers some flexibility when deciding the size of an ORAM blocks and the size of height of the ORAM tree. 
These two factors affect the size of the position map (resp. number of recursive levels) for non-recursive (resp. recursive) ORAM constructions. 	
However, since the output domain of $\mathsf{OBlockMap}(\cdot, \cdot)$ is limited to the size of the ORAM blocks, there will exist collisions. 
The following claim gives us a loose upper bound on the number of addresses that should be stored inside an ORAM block.
\begin{claim}(Addresses per ORAM block)\label{claim:addressesperoramblock}
	Let $m$ be the number of public key hashes,  $N$ be the number of ORAM blocks. If the $\mathsf{OBlockMap}()$ acts as a truly random function, then the maximum number of addresses in each ORAM block is smaller than $e\cdot m/N$ with a probability $1-1/N$.
\end{claim}
\begin{proof}
	This is a standard max-load analysis. 
	We refer readers to \cite{czumaj-bin-ball} for detailed analysis.
	We note that there exists a tighter bound, but we use $e\cdot m/N$ bounds to simplify the equation. 
\end{proof}

Thus, if we limit each ORAM block to contain the outputs of at most $e\cdot m/N$ addresses, then the probability that every address is included is at least $1-1/N$. \Cref{fig:single-add-to-single-block} gives us a high level overview of this approach.
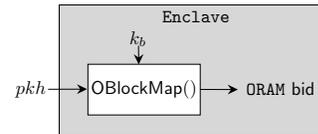
\begin{figure}[h!]
	\centering
	\resizebox{.50\linewidth}{!}{
\tikzset{every picture/.style={line width=0.75pt}} 

\begin{tikzpicture}[x=0.75pt,y=0.75pt,yscale=-1,xscale=1]

\draw  [fill={rgb, 255:red, 215; green, 215; blue, 215 }  ,fill opacity=1 ] (140,61) -- (418.5,61) -- (418.5,199) -- (140,199) -- cycle ;
\draw    (128.42,149.13) -- (167.17,149.01) ;
\draw [shift={(169.17,149)}, rotate = 539.8199999999999] [fill={rgb, 255:red, 0; green, 0; blue, 0 }  ][line width=0.75]  [draw opacity=0] (10.72,-5.15) -- (0,0) -- (10.72,5.15) -- (7.12,0) -- cycle    ;

\draw  [fill={rgb, 255:red, 255; green, 255; blue, 255 }  ,fill opacity=1 ] (170,123) -- (287.33,123) -- (287.33,176) -- (170,176) -- cycle ;
\draw    (222.83,103.67) -- (223.06,121.17) ;
\draw [shift={(223.08,123.17)}, rotate = 269.27] [fill={rgb, 255:red, 0; green, 0; blue, 0 }  ][line width=0.75]  [draw opacity=0] (10.72,-5.15) -- (0,0) -- (10.72,5.15) -- (7.12,0) -- cycle    ;

\draw    (287.33,149.33) -- (324.67,149.44) ;
\draw [shift={(326.67,149.44)}, rotate = 180.16] [fill={rgb, 255:red, 0; green, 0; blue, 0 }  ][line width=0.75]  [draw opacity=0] (10.72,-5.15) -- (0,0) -- (10.72,5.15) -- (7.12,0) -- cycle    ;

\draw (283,73) node [scale=1.44] [align=left] {$\displaystyle \mathtt{Enclave}$};
\draw (110.33,149.33) node [scale=1.44] [align=left] {$\displaystyle pkh$};
\draw (222.33,96.33) node [scale=1.44]  {$k_{b}$};
\draw (228.67,149.5) node [scale=1.44] [align=left] {$\displaystyle \mathsf{OBlockMap()}$};
\draw (371.33,148.67) node [scale=1.44] [align=left] {$\displaystyle \mathtt{ORAM}$ $\displaystyle \mathsf{bid}$};

\end{tikzpicture}
	}
	\caption{Single address into Single ORAM block}
	\label{fig:single-add-to-single-block}
\end{figure}


\noindent\textbf{Single address into Many ORAM blocks.}
Mapping a single address into a single ORAM block incurs less work on the server as it requires a single ORAM access for an address.
However, if we want to allow each address to have more than one output, 
using the first approach implies that the storage overhead increase linearly 
because the first approach distribute unspent outputs based on its addresses. 
Therefore, we have to pad dummy data for those addresses that contain $1$ outputs. 
Thus, we need a different mapping without linear increasing in storage overhead. 
To fix this shortcoming, the system needs to assign unspent outputs into ORAM block uniformly.  
One method is to allow a client to specify the number of ORAM accesses to obtain all of its unspent outputs as long as the number of requests does not exceed certain threshold. 
We define the mapping as follows:
\begin{compactitem}
	\item $ \{bid_i\}_{i \in \{0,\dots, \delta-1 \}} \leftarrow \mathsf{OBlockMap}(pkh, k_b, \delta)$: the function takes as input a $20$-bytes hash digest ${pkh}$, a secret key $k_b$, and a number $\delta$ where the maximum value of $\delta$ is specified by the system. It outputs a set of block identification numbers $ \{bid_i\}_{i \in \{0,\dots, \delta-1 \}} \subseteq \{0,\dots,N-1\}$.  
\end{compactitem}
This approach also introduces some leakage as some addresses may contain more unspent outputs than others. 
Alternatively, the system can fix the value of $\delta$ ORAM accesses for all addresses with the expense of performance (i.e., one address incurs constant ORAM accesses). Similarly, the storage overhead of \sys can be computed using the following claim:
\begin{claim}(UTXO per ORAM block)\label{claim:utxoperoramblock}
	Let $m$ be the number of unspent outputs, $N$ be the number of ORAM blocks. If the $\mathsf{OBlockMap}$ acts as a truly random function, then the maximum number of outputs in each ORAM block is smaller than $e\cdot m/N$ with probability at least $1 - 1/N$
\end{claim}
The proof is identical to proof of claim~\ref{claim:addressesperoramblock}. 

Figure~\ref{fig:single-add-to-many-block-fixed-not-fixed} offers an overview of the both approaches.
\begin{figure}[h]
	\centering
	\resizebox{\linewidth}{!}{
\tikzset{every picture/.style={line width=0.75pt}} 

\begin{tikzpicture}[x=0.75pt,y=0.75pt,yscale=-1,xscale=1]

\draw  [fill={rgb, 255:red, 215; green, 215; blue, 215 }  ,fill opacity=1 ] (55,104) -- (333.5,104) -- (333.5,242) -- (55,242) -- cycle ;
\draw    (44.75,178.13) -- (83.5,178.01) ;
\draw [shift={(85.5,178)}, rotate = 539.8199999999999] [fill={rgb, 255:red, 0; green, 0; blue, 0 }  ][line width=0.75]  [draw opacity=0] (10.72,-5.15) -- (0,0) -- (10.72,5.15) -- (7.12,0) -- cycle    ;

\draw  [fill={rgb, 255:red, 255; green, 255; blue, 255 }  ,fill opacity=1 ] (85,166) -- (202.33,166) -- (202.33,219) -- (85,219) -- cycle ;
\draw    (113.17,146) -- (113.39,163.5) ;
\draw [shift={(113.42,165.5)}, rotate = 269.27] [fill={rgb, 255:red, 0; green, 0; blue, 0 }  ][line width=0.75]  [draw opacity=0] (10.72,-5.15) -- (0,0) -- (10.72,5.15) -- (7.12,0) -- cycle    ;

\draw    (202.33,192.33) -- (238.66,176.79) ;
\draw [shift={(240.5,176)}, rotate = 516.8299999999999] [fill={rgb, 255:red, 0; green, 0; blue, 0 }  ][line width=0.75]  [draw opacity=0] (10.72,-5.15) -- (0,0) -- (10.72,5.15) -- (7.12,0) -- cycle    ;

\draw    (202.33,192.33) -- (237.56,201.5) ;
\draw [shift={(239.5,202)}, rotate = 194.58] [fill={rgb, 255:red, 0; green, 0; blue, 0 }  ][line width=0.75]  [draw opacity=0] (10.72,-5.15) -- (0,0) -- (10.72,5.15) -- (7.12,0) -- cycle    ;

\draw    (44.5,203) -- (83.5,203) ;
\draw [shift={(85.5,203)}, rotate = 180] [fill={rgb, 255:red, 0; green, 0; blue, 0 }  ][line width=0.75]  [draw opacity=0] (10.72,-5.15) -- (0,0) -- (10.72,5.15) -- (7.12,0) -- cycle    ;

\draw  [fill={rgb, 255:red, 215; green, 215; blue, 215 }  ,fill opacity=1 ] (376,104) -- (654.5,104) -- (654.5,242) -- (376,242) -- cycle ;
\draw  [fill={rgb, 255:red, 255; green, 255; blue, 255 }  ,fill opacity=1 ] (406,166) -- (523.33,166) -- (523.33,219) -- (406,219) -- cycle ;
\draw    (523.33,192.33) -- (559.66,176.79) ;
\draw [shift={(561.5,176)}, rotate = 516.8299999999999] [fill={rgb, 255:red, 0; green, 0; blue, 0 }  ][line width=0.75]  [draw opacity=0] (10.72,-5.15) -- (0,0) -- (10.72,5.15) -- (7.12,0) -- cycle    ;

\draw    (523.33,192.33) -- (558.56,201.5) ;
\draw [shift={(560.5,202)}, rotate = 194.58] [fill={rgb, 255:red, 0; green, 0; blue, 0 }  ][line width=0.75]  [draw opacity=0] (10.72,-5.15) -- (0,0) -- (10.72,5.15) -- (7.12,0) -- cycle    ;

\draw    (366.75,186.75) -- (403.5,186.99) ;
\draw [shift={(405.5,187)}, rotate = 180.37] [fill={rgb, 255:red, 0; green, 0; blue, 0 }  ][line width=0.75]  [draw opacity=0] (10.72,-5.15) -- (0,0) -- (10.72,5.15) -- (7.12,0) -- cycle    ;

\draw    (443.17,146) -- (443.39,163.5) ;
\draw [shift={(443.42,165.5)}, rotate = 269.27] [fill={rgb, 255:red, 0; green, 0; blue, 0 }  ][line width=0.75]  [draw opacity=0] (10.72,-5.15) -- (0,0) -- (10.72,5.15) -- (7.12,0) -- cycle    ;

\draw    (481.17,147) -- (481.39,164.5) ;
\draw [shift={(481.42,166.5)}, rotate = 269.27] [fill={rgb, 255:red, 0; green, 0; blue, 0 }  ][line width=0.75]  [draw opacity=0] (10.72,-5.15) -- (0,0) -- (10.72,5.15) -- (7.12,0) -- cycle    ;

\draw    (171.17,146) -- (171.39,163.5) ;
\draw [shift={(171.42,165.5)}, rotate = 269.27] [fill={rgb, 255:red, 0; green, 0; blue, 0 }  ][line width=0.75]  [draw opacity=0] (10.72,-5.15) -- (0,0) -- (10.72,5.15) -- (7.12,0) -- cycle    ;

\draw (200,114) node [scale=1.44] [align=left] {$\displaystyle \mathtt{Enclave}$};
\draw (28.67,176.33) node [scale=1.44] [align=left] {$\displaystyle pkh$};
\draw (113.67,136.67) node [scale=1.44]  {$k_{b}$};
\draw (143.67,192.5) node [scale=1.2] [align=left] {$\displaystyle \mathsf{OBlockMap()}$};
\draw (286.67,174.33) node [scale=1.44] [align=left] {$\displaystyle \mathtt{ORAM}$ $\displaystyle \mathsf{bid}$};
\draw (20.67,201.83) node [scale=1.44] [align=left] {$\displaystyle \delta =2$};
\draw (286.67,198.33) node [scale=1.44] [align=left] {$\displaystyle \mathtt{ORAM}$ $\displaystyle \mathsf{bid}$};
\draw (351.67,186.33) node [scale=1.44] [align=left] {$\displaystyle pkh$};
\draw (519,113.5) node [scale=1.44] [align=left] {$\displaystyle \mathtt{Enclave}$};
\draw (443.67,137.67) node [scale=1.44]  {$k_{b}$};
\draw (464.67,192.5) node [scale=1.2] [align=left] {$\displaystyle \mathsf{OBlockMap()}$};
\draw (606.67,174.33) node [scale=1.44] [align=left] {$\displaystyle \mathtt{ORAM}$ $\displaystyle \mathsf{bid}$};
\draw (606.67,198.33) node [scale=1.44] [align=left] {$\displaystyle \mathtt{ORAM}$ $\displaystyle \mathsf{bid}$};
\draw (484.67,136.67) node [scale=1.44]  {$\delta =2$};
\draw (173.67,134.67) node [scale=1.44]  {$max=3$};

\end{tikzpicture}
	}
	\caption{Single Address into Many ORAM blocks}
	\label{fig:single-add-to-many-block-fixed-not-fixed}
\end{figure}
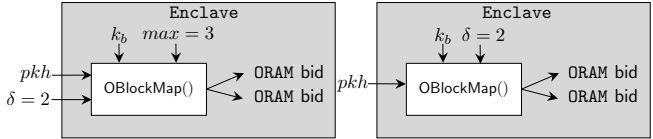
\subsubsection{Storage}\
In this system, we require the untrusted server to store three separate databases which are the \readTree, the \updateTree, and the blockheader chain. In particular,
	\textbf{\textit{Read-Once} ORAM Tree} 
	serves as a dedicated storage to handle clients' requests. 
	The structure of the tree is identical to the standard ORAM tree
	\textbf{\textit{Original} ORAM Tree} is where all standard ORAM eviction operations are performed. 
	In this work, we also require the enclave to maintain the \textbf{Bitcoin Header Chain} to verify the proof of work of the bitcoin block sent by other bitcoin client. The header chain is stored in the untrusted memory with integrity check.
\subsection{Oblivious Read and Write Protocols}
\label{sub:oRAM operations}
In \SystemName, the SPV client is the party who invokes read accesses, and the Bitcoin network is the party who invokes write accesses. 
The TEE in the server is the one that performs both of those accesses on behalf of the client and the Bitcoin network. 
\subsubsection{Server System Components}
Before explaining how oblivious read and write accesses work, we first start outlining the different components of our design.
The server is initialized with different enclaves:
\paragraph{Managing Enclave $\mathcal E_{m}$} 
coordinates other enclaves and to handle requests from the clients. 
The \textit{managing} enclave also handles the communication with other Bitcoin client or local Bitcoin client ($\mathsf{bitcoind}$) via request procedure calls (RPC) to obtain Bitcoin blocks. 
Upon receiving the Bitcoin block,  the \managingEnclave also verifies the integrity of the block using a separated header chain.
\paragraph{Reading Enclave $\mathcal E_{r}$} is a dedicated enclave initialized by the \managingEnclave. 
It has  a copy of ORAM position map and its own stash. The \readEnclave operates on the \readTree. 
Also, the \readEnclave only performs ORAM $\mathsf{ReadPath}$ operations to obtain data while ORAM $\mathsf{Eviction}$ operations will be handled by the \updateEnclave.
\paragraph{Writing Enclave $\mathcal E_{w}$} performs $\mathsf{Eviction}$ procedure for each read request, and performs ORAM writing accesses when a new Bitcoin block arrives from the Bitcoin network.

\subsubsection{Oblivious \readname Protocol}
\label{subsub:read-proc}
Here, we describe how a remote client can perform a read access on the UTXO set. 

\paragraph{Notation. } 
First, let's denote 
$K_b$ to be the block mapping key, 
$\mathsf{bid}$ to be the ORAM block identification. 
We let $(\mathsf{Enc}, \mathsf{Dec})$ denote an authenticated encryption scheme.
We assume that the the server has already been initialized with a \updateEnclave, \uEnclave and a \managingEnclave, \mEnclave. 
The \managingEnclave has a similar copy of the position map as the map in the \updateEnclave.
Figure~\ref{fig:read-overview} presents the oblivious read protocol of \SystemName. The oblivious read protocol can be described as follows:
\begin{figure}[t]
\centering
\includegraphics[width=.95\columnwidth]{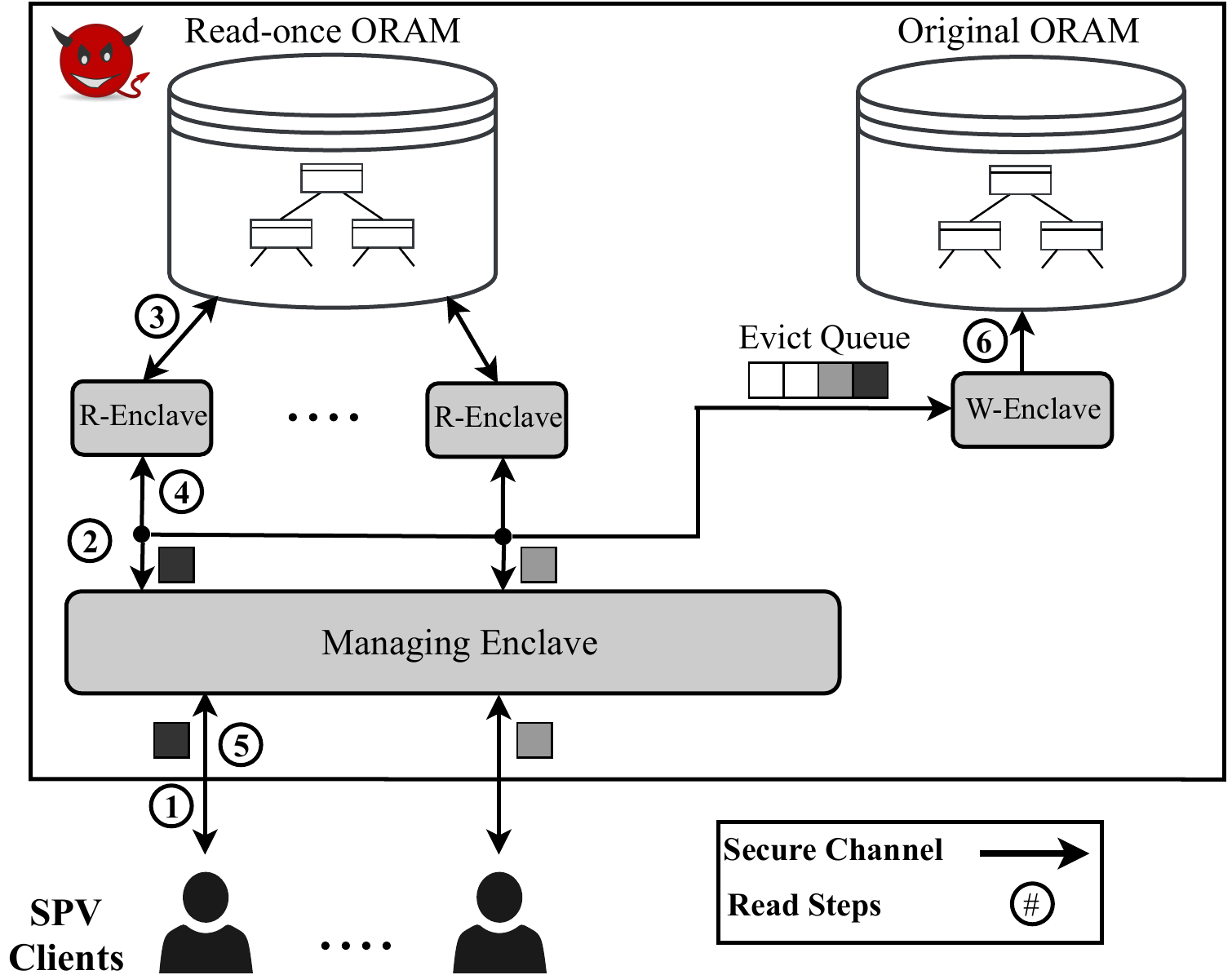}
\caption{The read protocol.
Steps \protect\bcircled{1}-\protect\bcircled{5} describes how \protect\sys receives and responds to the client, and for each request, the \updateEnclave performs the $\mathsf{Eviction}$ procedure of ORAM on the \updateTree during step \protect\bcircled{6}.
}
\label{fig:read-overview}
\end{figure}
\begin{asparaenum}
	\item \textbf{The client establishes a secure channel with the \managingEnclave}~\bcircled{1}: First, the client performs a remote attestation with the secure \textit{managing} enclave, \mEnclave , and agrees on a session key, $K_s$. 
	The client encrypts his address along with the proof of ownership of that address, and sends the encrypted query to the server to be passed to \mEnclave. For simplicity, we assume that the plaintext only contains a public key hash, $pkh$, that the client is interested in, and the proof of ownership of the $pkh$ is $\phi$, $C\leftarrow\mathsf{Enc}_{K_s}(pkh, \phi)$. 
	Note that there are different ways to prove the ownership of public key hash/addresses. 
	In Bitcoin, if the public key is never revealed before, the proof of ownership can simply be the public key (i.e. $\phi = pk$ such that $H(pk)=pkh$).
	Alternatively, the system can enforce a client to provide the signature and the public key to prove the ownership of the public key hash.

	\item \textbf{The \textit{managing} enclave initializes a \textit{reading} enclave}~\bcircled{2}~: 
	after receiving a client's request, \mEnclave initializes a dedicated \readEnclave, \rEnclave to handle the client's future requests. 
	Also, we require that the enclaves authenticate each other, and the existence of a secure channel between enclaves. 
	Moreover, the \readEnclave has its copy of the position map, its own stash, the block mapping key $K_b$, and the agreed session key $K_s$.

	\item \textbf{The \textit{managing} enclave identifies and forwards ORAM Block ID to both \textit{reading} and \textit{writing} enclaves}~\bcircled{2}: 
	After decrypting the ciphertext $(pkh, \phi) \leftarrow \mathsf{Dec}_{K_s}(C)$, \mEnclave verifies the proof $\phi$ and $pkh$,
	then uses $\mathsf{OBlockMap}(\cdot,\cdot)$~\footnote{for simplicity, we assume that the one-to-one mapping is used here} function to learn the ORAM block ID, $\mathsf{bid} \leftarrow \mathsf{OBlockMap}(pkh, K_b)$ where $K_b$ is the secret key generated by the enclave during initialization for mapping purposes. 
	After obtaining the ORAM id, $\mathsf{bid}$, the \textit{managing} enclave forwards $\mathsf{bid}$ to the \updateEnclave for the eviction procedure, 
	and forwards the $(pkh,\mathsf{bid})$ to the \readEnclave.

	\item \textbf{The \readEnclave performs \readaccessname on the \readTree}~\bcircled{3}: Based on the given $\mathsf{bid}$, the \readEnclave performs ORAM read only accesses on the ORAM tree to obtain the block. 
	If the block contains the unspent output that belongs to the public key $pkh$, 
    the \readEnclave adds outputs into the response $R$. 
	To mitigate the size leakage, the response $R$ is padded with dummy data if there is no UTXO found.
	\item \textbf{The \readEnclave responds to the Client}~\bcircled{4}-\bcircled{5}~: The enclave encrypts the response, $R$, using the session key $K_s$ then sends it to the client. 
	\item \textbf{The \updateEnclave performs the eviction procedure on the \updateTree}~\bcircled{6}: After obtaining the $\mathsf{bid}$ from the \managingEnclave, the update enclave will perform a standard ORAM read accesses on the \updateTree. 
	The goal of this procedure is to use the \evict~procedure inside standard ORAM operation to rerandomize the location of the actual block. No actual data is return in this step.
	\end{asparaenum}

\subsubsection{Oblivious Write Protocol}
\label{subsub:write-proc}
In the Bitcoin network, miners generate a new Bitcoin block on average every 10 minutes. 
When the server receives a new block from the Bitcoin network, the \managingEnclave can obtain it from the bitcoin client. 
The \mEnclave verifies the integrity of the block by computing the Merkle root from transactions, then verifying the proof of work, 
and  the \updateEnclave has to perform an update on the \updateTree; 
however, in the mean time, the system should be able to handle clients' requests during updates.
We will explain how \SystemName handles oblivious write accesses while handling clients' requests as follow:



\begin{asparaenum}
	\item \textbf{The \textit{managing} enclave verifies a new Bitcoin block}~\rcircled{1}-\rcircled{3}~: Once a bitcoin block arrives to the system from the Bitcoin network, the \managingEnclave \mEnclave can obtain it from the Bitcoin client. 
	However, since the client runs outside the enclave, the enclave needs to verify the integrity of the new block by computing the Merkle root and verifying the proof of work to make sure that the block has not been tampered by the untrusted OS. 
	For the detail of these computations, we refer readers to \cite{btc-reference}. 
	Moreover, as discussed in section~\ref{subsec: Oblivious Storage of the UTXO set}, to verify a newly arrived block, the system is required to keep a separate block headers chain with integrity check in the untrusted memory. 
	Once \mEnclave verifies the bitcoin block, \mEnclave starts pruning the transactions to obtain relevant information of the transactions' inputs and outputs. 
	Then, \mEnclave uses $\mathsf{OBlockMap}(\cdot, \cdot)$ to find the ORAM block identification to queue up ORAM write requests to the \updateEnclave. 
	During this process, the oblivious read protocol performs as normal on the \readTree. 
\begin{figure}[t]
	\centering
	\includegraphics[width=.95\columnwidth]{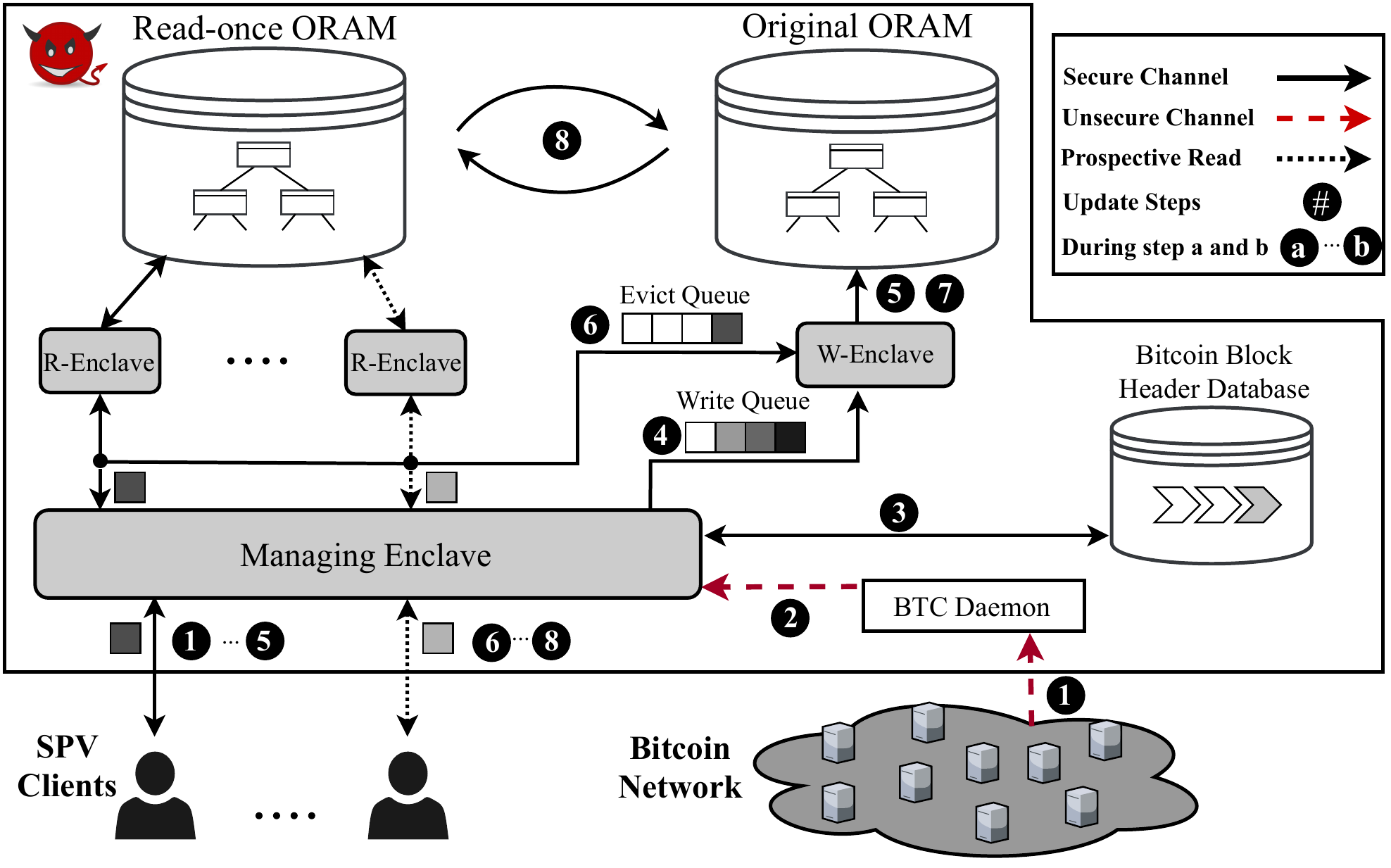}
	\caption{Oblivious write protocol. 
	During steps \protect\rcircled{1}-\protect\rcircled{5}, the \managingEnclave receives and responds to SPV client request as usual. During steps \protect\rcircled{6}-\protect\rcircled{8}, read requests from clients are queued up, and the \managingEnclave resume these requests after updating the \readTree
	}
	\label{fig:write-overview}
\end{figure}
	\item \textbf{The \textit{managing} enclave sends write requests to the \updateEnclave}~\rcircled{4}: 
		Once the pruning process completes, the \mEnclave starts sending write requests based on data extracted from the bitcoin block to the \updateEnclave, \uEnclave. 
		On otherhand, for each eviction request resulted from SPV client's requests, \mEnclave starts queuing up those eviction requests. 

	\item \textbf{The \textit{writing} enclave performs write accesses on the \updateTree}~\rcircled{4}-\rcircled{5}: 
		Upon receiving writing requests from \mEnclave, the \uEnclave performs all writing requests in the writing queue on the \updateTree. 

	\item \textbf{The \updateEnclave finishes all eviction requests queued up on the \updateTree}~\rcircled{6}-\rcircled{7}:
		Once finished updating the tree, the \uEnclave signals \mEnclave to start queuing up clients' requests and performs all eviction requests incurred by SPV clients' read requests during update interval.
		Finally, when it finishes, it signals the \mEnclave to update the \readTree and make a copy of the position map. 
	\item \textbf{The \textit{managing} enclave performs an update the \readTree and the \updateTree and enclave metadata}~\rcircled{8}: In particular, \mEnclave discards the current copy of the \readTree, and makes 2 identical copies of the most updated \updateTree. One is used as \readTree, and the other is used as \updateTree. Also, the new position map and new stash are updated for the \managingEnclave. 
	Once this process is finished, \mEnclave starts answering SPV clients' requests again.
\end{asparaenum}
Figure~\ref{fig:write-overview} gives us an overview of the oblivious write protocol.

\paragraph{Discussion}
The core idea of the oblivious update protocol is to minimize the downtime of \SystemName during the update process. More specifically, during step~\rcircled{1}-\rcircled{5}, \SystemName still allows SPV clients to query the system while from step~\rcircled{6}-\rcircled{8}, \SystemName stops accepting clients' requests in order to synchronize both trees. This approach introduces some delay; however, the system downtime is minimized to the same amount of time it takes for the \updateEnclave to finish all eviction requests.

\begin{table*}[h]
\centering
\resizebox{.9\textwidth}{!}{%
\begin{tabular}{c c|c|c|c|c|}
\cline{3-6}
    &             & \multicolumn{2}{c|}{\textbf{\sys(\pathORAM, $Z=4$)} $$}           &       \multicolumn{2}{c|}{\textbf{\sys(\circuitORAM, $Z=2$)} }     \\ \hline
 \multicolumn{1}{|c|}{$N$} &  \multicolumn{1}{|c|}{Block Size}  & \readname Access & Standard ORAM access  & \readname Access & \multicolumn{1}{c|}{Standard ORAM access} \\  \hline
 \multicolumn{1}{|c|}{$2^{20}$} &  \multicolumn{1}{|c|}{6528 bytes (96 utxos)}  & $16.34$ ms & $30.40$ ms  & $2.13$ ms   & \multicolumn{1}{c|}{$6.45$ ms} \\  \hline
 \multicolumn{1}{|c|}{$2^{21}$} &  \multicolumn{1}{|c|}{3264 bytes (48 utxos) } & $9.24$  ms & $16.58$ ms  & $1.27$ ms   & \multicolumn{1}{c|}{$3.76$ ms} \\  \hline
 \multicolumn{1}{|c|}{$2^{22}$} &  \multicolumn{1}{|c|}{2176 bytes (32 utxos)}  & $7.56$  ms & $12.42$ ms  & $1.05$ ms   & \multicolumn{1}{c|}{$2.92$ ms} \\  \hline
 \multicolumn{1}{|c|}{$2^{23}$} &  \multicolumn{1}{|c|}{1088 bytes (16 utxos)}  & $4.12$  ms & $7.78$ ms   & $0.72$ ms   & \multicolumn{1}{c|}{$2.09$ ms} \\  \hline
 \multicolumn{1}{|c|}{$2^{24}$} &  \multicolumn{1}{|c|}{544  bytes (8  utxos)}  & $2.43$  ms & $5.89$ ms   & $0.64$ ms   & \multicolumn{1}{c|}{$1.70$ ms} \\  \hline
\end{tabular}
}
\vspace{10pt}
\caption{Performance of two different types of \textsc{Path}/\circuitORAM accesses on different block size.}
\label{table:path-circuit-oram}
\end{table*}

\section{Evaluation and Comparison} 
\label{sec:Evaluation}

In this section, we describe our configuration, our experimental results, and the storage overhead of the system based on the analysis of the UTXO set on the Bitcoin blockchain. Moreover, we give a comparison between \sys and the current existing SPV solution in term of performance and communication overhead. Finally, we address the capabilities of \sys compared to other related works. 
\subsection{Configuration} 
\label{sub:configuration}
\paragraph{Software. } 
We implemented our system with C++ using Intel SGX SDK v2.0.
The implementation of the ORAM controller is built on top the $\mathsf{Zerotrace}$~\cite{SasyGF18-zero-trace} implementation. In order to handle the communication with the Bitcoin network, 
we have used \texttt{libjson-rpc-cpp}~\cite{libjson-rpc} framework to build C++ wrapper functions to communicate with the Bitcoin daemon (\texttt{bitcoind}~\cite{bitcoin-core}) from inside the enclave through JSON-RPC calls. For extracting the UTXO database, we used the \texttt{bitcoin-tool} implementation proposed in~\cite{analysis-of-utxo}. 
This allows us to save time during the initialization phase. Finally, we used \texttt{python-bitcoinlib}~\cite{pythonbitcoin} to compare the performance of \sys with the current existing SPV solution.

\paragraph{Database. } To reduce the time of initializing both ORAM trees from the \textit{genesis} block, we used $\texttt{bitcoin\text{-}tool}$ implementation proposed in~\cite{analysis-of-utxo} to extract $3.2$GB of the Bitcoin {UTXO} set in February 2019. 

We have downloaded a snapshot of the Bitcoin blockchain including block 0 to $551,731$, containing a total of $58,156,895$ Unspent Transaction Outputs (UTXO). 
\Cref{fig:utxo-analysis} shows the distribution of the unspent transaction outputs per address. 
We see that more than 90\% of the addresses have less than three UTXOs. 
In our prototype, we considered at most two UTXOs per wallet ID. 
This results in covering more than 92\% of all the UTXOs per wallet ID.
Also, as discussed in section~\ref{sec:new-protocol}, by using different mapping, one can cover more percentage of Bitcoin addresses.
\begin{figure}[b]
	\centering
	\includegraphics[width=.9\linewidth]{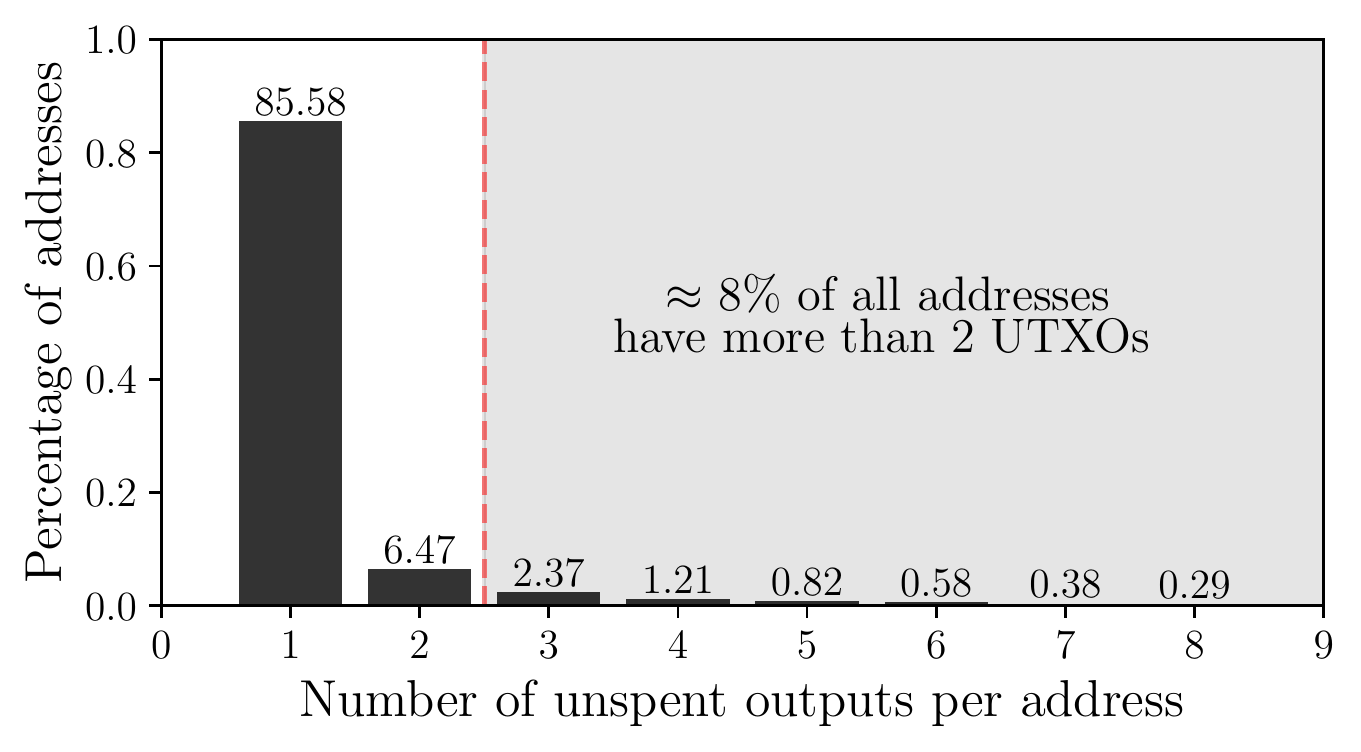}
	\caption{Number of transactions per wallet ID. By allowing each address can have up to 2 UTXO, \sys can cover approximate $92\%$ of the UTXO set.} 
	\label{fig:utxo-analysis}
\end{figure}

\paragraph{Hardware.}
We evaluated the performance of \SystemName on a desktop which is equipped with Intel(R) Xeon(R) Silver 4116 CPU @ 2.10GHz, 128GB RAM.
Since Intel(R) Xeon(R) silver 4116 is not SGX-enabled CPU, we obtain the performance results by running our implementation in the simulation mode. 
However, we expect to not have much of a performance difference when executing in the two different modes.
More specifically, we have tested the performance of \sys 
using a smaller ORAM tree 
in the hardware mode on a commodity desktop equipped with SGX-enabled Intel Core i7.
Comparing the hardware and simulation mode results (i.e., simulation on the Intel Core i7 CPU),
we see no noticeable difference in the running time of both \readname and standard ORAM accesses.

\subsection{Experimental Results} 
\label{sub:perforamce_analysis}
We have implemented \sys using multiple threads. 
As reported in \cite{multiple-thread-sgx,tramer2018slalom}, as long as the total amount of memory used by all threads does not exceed the EPC limit, 
the performance gain should be similar to the use of different enclaves. 
In this work, we implemented all functionalities in one single enclave, and we used multiple threads to concurrently accesses the ORAM trees.

\paragraph{System parameters}
We tested our system with both recursive \pathORAM~and recursive \circuitORAM~using different tree size $N=2^{20}, 2^{21}, 2^{22}, 2^{23}, 2^{24}$. 
We allow each Bitcoin address to have up to $2$ unspent transaction outputs, 
and we use the single address into single ORAM block mapping approach described in section~\ref{subsec: Oblivious Storage of the UTXO set} to map addresses into ORAM block. 
Finally,  we use claim~\ref{claim:addressesperoramblock} to determine the size of each ORAM block.

\paragraph{Performance of \readname and standard ORAM accesses.} 
In $T^3$, the \readEnclave performs \readname accesses to handle client's requests in an efficient manner. 
Table~\ref{table:path-circuit-oram} presents an overall performance of a standard ORAM access as well as the performance of a \readname access for both \circuitORAM~and \pathORAM. 
For this experiment, we took the average running time of 10000 accesses.


As shown in the results,
ORAM constructions with smaller block sizes provides a  better performance in both schemes.
The reason is that oblivious operations like oblivious comparisons and $\mathsf{cmov}$-based stash scan are more efficient because of a smaller size stash. 
Moreover, \circuitORAM gives a better performance compared to \pathORAM, as it can operate on a smaller block compared to \pathORAM, and this requires much smaller stash size
allowing much faster oblivious execution.
\begin{table}[b]
\centering
\resizebox{.95\columnwidth}{!}{%
\begin{tabular}{c|c|c|}
\cline{2-3}
                                                            & \multicolumn{1}{c|}{\multirow{2}{*}{\sys(\textbf{\pathORAM})}}           &       \multicolumn{1}{c|}{\multirow{2}{*}{\sys(\textbf{\circuitORAM})}}\\\cline{1-1}
 \multicolumn{1}{|c|}{\textbf{Number of threads}}                    &   			 &  \multicolumn{1}{c|}{}          		\\  \hline
 \multicolumn{1}{|c|}{$1$}	 								& {2.43} ms & \multicolumn{1}{c|}{0.64 ms}   \\  \hline
 \multicolumn{1}{|c|}{$2$} 									& {1.40} ms & \multicolumn{1}{c|}{0.58 ms}  	\\  \hline
 \multicolumn{1}{|c|}{$3$} 									& {0.90} ms & \multicolumn{1}{c|}{0.43 ms}  	\\  \hline 
 \multicolumn{1}{|c|}{$4$} 									& {0.73} ms & \multicolumn{1}{c|}{0.35 ms}  	\\  \hline
\end{tabular}
}
\vspace{10pt}
\captionof{table}{Performance gain of multiple-thread \readname access on Path/\circuitORAM with $N=2^{24}$ block size = $544$ bytes.}
\label{table:path-circuit-oram-multiple-threading}
\end{table}

\paragraph{Parallelization.} Since there is no race condition in the \readname accesses, the design of \SystemName allows different threads to concurrently perform \readname accesses on the \readTree. 
Compared to other oblivious system like \textsc{Bite}~\cite{SasyGF18-zero-trace}, \SystemName is able to handle bursty client read requests concurrently while the eviction requests are distributed sequentially during the \bci. 
To measure this performance gain, we used multiple threads to access the \readname enclave and perform \readname access simultaneously on a tree of size $N=2^{24}$ and ORAM block of size $544$ bytes. 
Table~\ref{table:path-circuit-oram-multiple-threading} shows the performance of \sys implemented using multiple threads for both \circuitORAM~and \pathORAM.


\paragraph{Comparison to current SPV solutions. }
We give a comparison in term of performance and communication overhead over several number of requests to the existing SPV client's solution and to BITE~\cite{bite-spv-sgx}
Oblivious database.
\begin{asparaenum}
	\item \textit{Performance}:
	\Cref{fig:performance} gives us an overview of the performance of \sys compared to the performance of the current existing SPV with Bloom filter solution and the performance of Bite Oblivious database. 
	In particular, it shows the response latency from the client's perspective. 
	In this comparison, a request for the SPV solution with Bloom filter solution means the time the server takes to scan one Bitcoin block, 
	and a request for \sys and BITE means the time it takes to perform an ORAM access on the ORAM tree. 
		For the current SPV clients with Bloom filter, we set the false positive rate of the Bloom filter to $1.0$\% and $5.0$\% respectively. 
	For \sys, we used $N=2^{24}$ and block of size 544 bytes for both \pathORAM with $Z=4$ and \circuitORAM with $Z=2$.
	For BITE database, based on our understand of their construction,  we re-implemented BITE using non-recursive construction of \pathORAM, and we used the same ORAM block of size $32$kB which leads to the number of block is $N=2^{17}$. 
	Also, we also provide an additional construction of BITE which is implemented using recursive \pathORAM and suggested parameters for \sys where the tree is of size $2^{24}$ and block of size $544$B. \Cref{fig:performance} gives us the overall performance of three existing solutions. 
	
	The performance of \sys outperforms the SPV with Bloom filter solution. 
	The reason is that in \sys, the system relies on the TEE to handle the integrity checking of the Bitcoin block before updating the ORAM tree while in the current SPV solution, the full client needs to scan the Bloom filter every time and detect the relevant transactions and recompute the Merkle path for each of those transactions. 
	
	Also, \sys performs much better than BITE oblivious database as the BITE system does not consider the use of recursive ORAM construction. 
	Another reason is that the size of the ORAM block used in Bite is large; hence, the cost of oblivious operation like $\mathsf{cmov}$-based stash scan becomes more expensive.
	Thus, we envision and realize an improved construction of BITE using recursive construction of \pathORAM to demonstrate the practical impact of using recursive ORAM construction on TEE with restricted memories. 
	\begin{figure}[b]
		\centering
		\includegraphics[width=\linewidth]{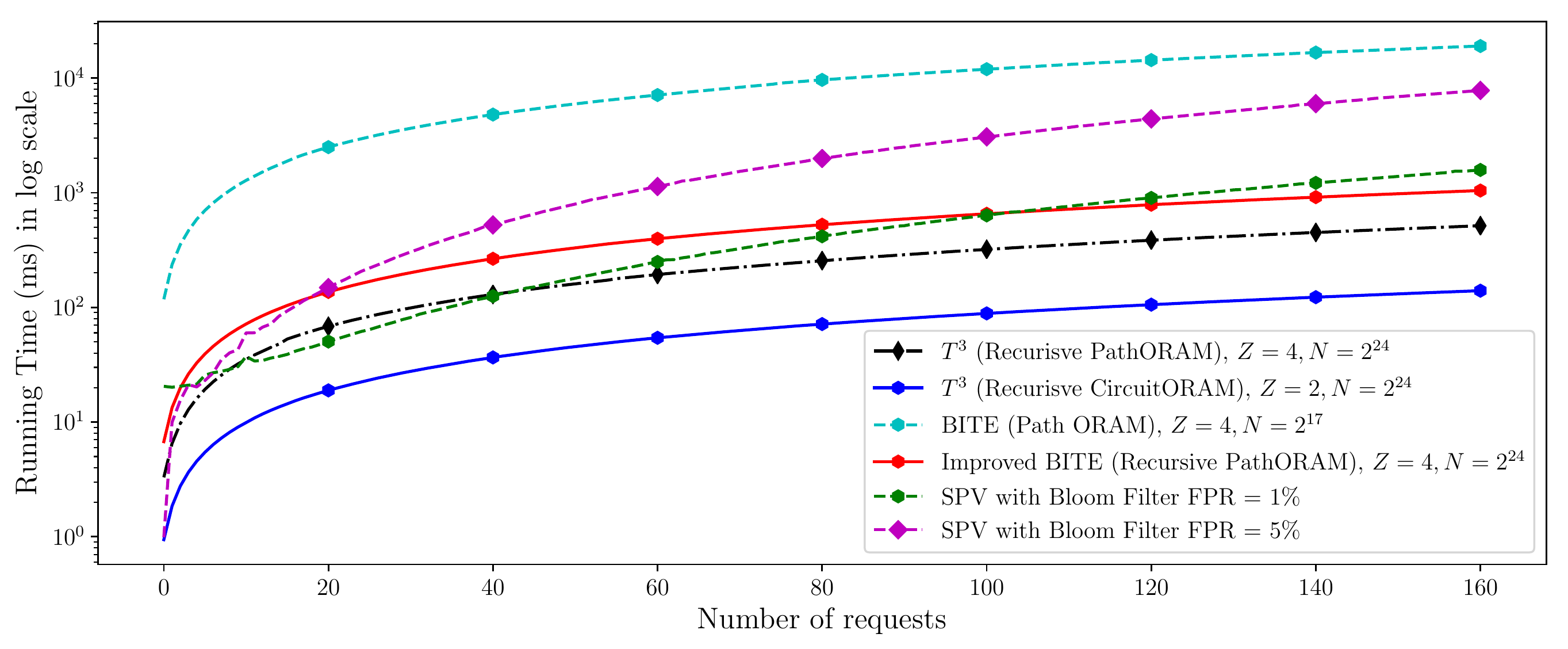}

		\caption{Performance of $T^3$ using \textsc{Path}/\textsc{Cicruit} ORAM with block of size $544$B, the current SPV with Bloom filter, Original BITE oblivious database block of size $32$kB, and improved BITE with block of size $544$B. For the SPV client with Bloom filter, we used the false positive rate of $1\%$ and $5\%$.}
		\label{fig:performance}
	\end{figure}

	\item \textit{Communication Overhead}: 
	In term of communication between client and server, \sys offers much lower communication overhead compared to the existing solution for SPV clients. 
	\sys does not need to provide the SPV clients with the Merkle proofs to its relevant transactions because all those proofs are validated by the Intel SGX before being added the ORAM tree. 
	Thus, one can reduce both the amount of work that the full node needs to perform and the amount of data that it needs to send to the SPV clients. 
	Moreover, \sys prunes all other information of transactions to extract only relevant data needed for client to determine balance and form new transactions, 
	while in the current SPV solution, due to the false positive rate used in the Bloom filter, the full client may send additional irrelevant information to the SPV client.  
	~\Cref{fig:communication} shows an overview of the communication cost of \sys compared to the current solution. 
	To give an estimation of the communication cost of the current SPV solution, we assumed that each request requires a separate Merkle proof. 
	Hence, for each request, the size of the proof is at least: $\log_2(\textsf{NoTXs})\cdot 32$ bytes where $\mathsf{NoTXs}$ is the number of transactions in one block.
	Moreover, the size of the transaction data is approximately $\mathsf{fpr} \cdot \textsf{BlockSize}$~bytes where the $\mathsf{fpr}$ is the false positive rate and the $\mathsf{BlockSize}$ is the size of the Bitcoin block.
	To compute the overhead cost we used block 551731, as an example, which has block size of 1149 KB and contains 3017 transactions. 
	However, in practice, we would expect the Bitcoin blocks to have different sizes; resulting, the communication cost to be different across blocks.
	Therefore, the results in \cref{fig:communication} is only an estimation on the communication overhead using the current SPV solution. 
	We omit the comparison to the communication overhead of BITE because both \sys and BITE return a fixed amount a data to the SPV client which is the output itself. 
	\begin{figure}[t]
		\centering
		\includegraphics[width=\linewidth]{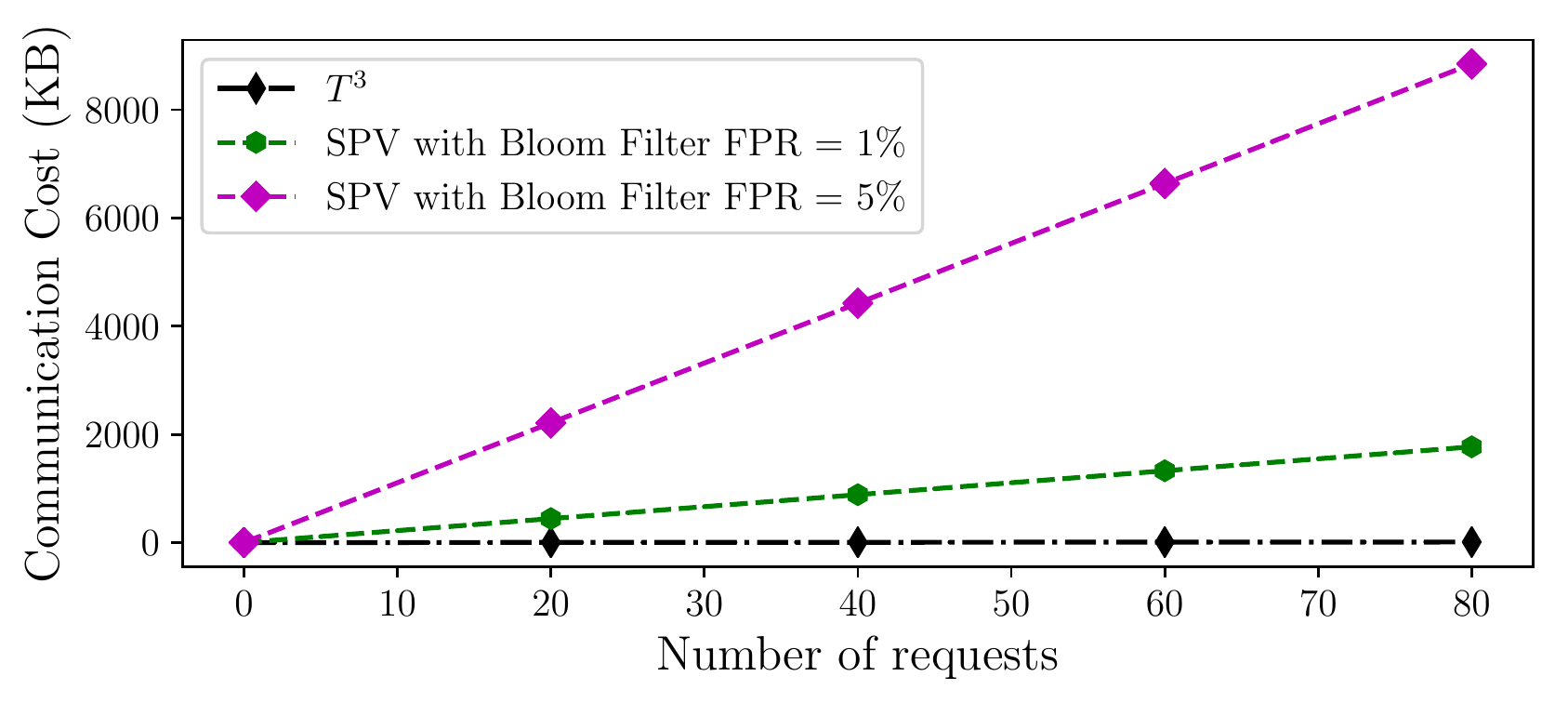}

		\caption{Communication cost of $T^3$ and the current SPV solution. Since both systems return the information of unspent outputs to the client, the communication overhead of BITE will be equal to the communication overhead of \sys.}   
		\label{fig:communication}
	\end{figure}
\end{asparaenum}

\paragraph{Storage Overhead}
As noted in the previous section, using ORAM incurs a constant size blow up of the storage of the UTXOs (e.g., $\approx 3-4\times$ for \circuitORAM, $6-8\times$ for \pathORAM). 
In particular, for \pathORAM with $Z=4$, the storage cost of ORAM trees is about $\approx 51GB$, 
and for \circuitORAM with $Z=2$, the storage cost of two ORAM tree is around $\approx 26GB$. 
For the EPC memory usage, we need to consider the size of the position map and the stash used by each enclave. 
In this work, since we use recursive ORAM constructions for both schemes, 
the size of the position map can be as small as possible at the cost of storing more recursive ORAM trees in the untrusted memory region. 
Precisely, in the prototype of our implementation, each thread uses $8$KB for the position map and a stash of size $2\cdot \log(N)\cdot Z \cdot \mathsf{BlockSize}$ bytes (e.g., for a tree of size $N=2^{24}$, we use $\approx0.62$MB bytes for \pathORAM, and $\approx 0.31$ MB for \circuitORAM). 
Thus, as long as the total memory usage by all threads/enclaves does not exceed the EPC limit (e.g., $96$MB), the performance of the system will not suffer from the expensive swapping operations discussed in section~\ref{sub:intel_sgx}. Moreover, for integrity protection, \sys only requires the server to store the Bitcoin header chain which is approximately $44$MB instead of storing a complete Bitcoin blockchain. Thus, in the future work, if \sys can handle the communication with the Bitcoin network without the reliance on the existing Bitcoin client, \sys reduces the need of storing the 230 GB of Bitcoin blockchain.

\subsection{Comparison with Other Oblivious Systems}
Here we provide a comparison between \sys and other generic oblivious systems. We compare our work with $\textsc{Bite}$~\cite{bite-spv-sgx} Oblivious Database 
that also uses ORAM and TEE to provide a generic PIR system for Bitcoin client, $\textsc{ConcurORAM}~\cite{concurORAM-chakraborti}$ that provides concurrency access to ORAM clients, $\textsc{Obliviate}~\cite{ndss-AhmadKSL18}$ that prevents leakage from file system accesses, and $\textsc{ZeroTrace}$ which proposes an efficient generic oblivious memory access primitives.
In particular, \cref{table:comparision} compares those systems based on the capabilities of supporting concurrency access, enabling recursive construction, and preventing side-channel leakage.

\begin{table}[b]
\centering
\begin{minipage}{\columnwidth}
\resizebox{.98\columnwidth}{!}{%
    \begin{tabular}{lccc}
      \toprule
      & \multicolumn{3}{c}{Capabilities}\\
      \cmidrule(lr){2-4}
      System    			                                & Concurrency & Recursive Construction & Side-channel Protection\\
      \midrule
      \textsc{ConcurORAM}~\cite{concurORAM-chakraborti}		& \cmark      & \xmark    & -~\footnote{\textsc{ConcurORAM} does not aim to provide side-channel protection for TEE. Hence, we omit this comparison.}           \\
      \textsc{Obliviate}~\cite{ndss-AhmadKSL18}    			& \xmark      & \xmark    & \cmark      \\
      \textsc{Zerotrace}~\cite{SasyGF18-zero-trace}    		& \xmark      & \cmark    & \cmark      \\
      \textsc{Bite} Oblivious Database~\cite{bite-spv-sgx} 	& \xmark      & \xmark    & \cmark      \\
      \sys 	        										& \cmark      & \cmark    & \cmark      \\
      \bottomrule
    \end{tabular}
}
\end{minipage}
\label{table:comparision}
\caption{Comparison between \sys and other oblivious systems.}
\end{table}

For generic trusted hardware-based systems like \textsc{Bite} oblivious database and \textsc{Obliviate}, while providing protection against side-channel leakage, those systems do not consider the use of recursive ORAM construction to reduce the EPC memory usage. 
Hence, the performance of their systems will degrade once the database becomes too large.  
Other works that harnesses the use of recursive ORAM construction are \textsc{Zerotrace}; however, concurrency is not supported in the current version of \textsc{Zerotrace}. Thus, without concurrency support, such systems will not scale well to handle Bitcoin SPV clients.
\textsc{ConcurORAM} is a recent ORAM construction that offers concurrency accesses from the clients; however, due to more optimized eviction strategy and complex synchronization schedule, the recursive construction of \textsc{ConcurORAM} introduces implementation challenges. 
Nevertheless, we believe that it can be an interesting future work to use \textsc{ConcurORAM} in the design of \sys. 
\section{System Analysis}
\label{sec:security-analysis}

\subsection{Security Claims}
\label{sub:privacy}
In order to prove the security properties of \sys's design, we put forth six
claims, each of which represents the security of a major component of \sys in term of privacy goal.

\paragraph{Claim 1. The managing enclave does not leak user-related information to an attacker.}
The managing enclave is responsible for three tasks ---
(a) converting wallet IDs to UTXOs, (b) creating and managing
threads which will perform read operations on the \readTree, 
and (c) handle the updates to be performed on the \updateTree.

Firstly, the conversion of wallet IDs to their respective UTXOs is
private since the channel between clients and the \managingEnclave is secured by the shared key during the remote attestation process.
More importantly, when receiving addresses from a client, the \managingEnclave uses blockmapping function (described in~\ref{subsubsec: btcintoORAM}) to map each address to a fixed number of ORAM blocks. 
This does not reveal information about the number of outputs belonging
to an address.
Secondly, each read thread performs the same operations irrespective
of the wallet ID provided to it, i.e., each thread simply retrieves
an ORAM block using ORAM accesses implemented with $\mathsf{cmov}$-based oblivious executions. 
Lastly, the only thing revealed by the update process of \sys is the
number of blocks updated into the Write Tree. However, this is public
information and \sys does not try to hide it.
Each update is performed using an ORAM access which ensures that the attacker is unaware of the final position of each block.

\paragraph{Claim 2. The optimized read operations on \readTree do not leak information.}
As explained in section~\ref{subsub:read-proc}, the \readTree is accessed using an optimized read
operation which chooses not to shuffle and write-back the retrieved path to the \readTree.
However, this is secure since each path corresponding to a UTXO can
only be accessed once during a read interval and will be shuffled
before the next interval.
%
%


\paragraph{Claim 3. The write operations performed on the \updateTree
do not leak information.}
There are two specific operations performed on the \updateTree  --- (a) the UTXOs
are updated based on the updated bitcoin block, and
(b) the previously accessed ORAM blocks are shuffled.
However, all of these updating accesses are standard ORAM operations implemented in a side-channel-resistant manners as previously done by~\cite{SasyGF18-zero-trace,ndss-AhmadKSL18}. Therefore, all write operations reveal no information about a user's UTXO.

\paragraph{Claim 4. The data fetched from the untrusted world to the TEE is correct.}
There are two major sources of data transferred from the untrusted to
the trusted world --- (a) the updated block fetched from the bitcoin
daemon after a fixed interval and (b) the ORAM tree blocks which are
fetched from the untrusted world into the TEE.
As mentioned in~\ref{subsub:write-proc}, Bitcoin blocks are fetched from outside the
enclave. However, \sys verifies the integrity of the Bitcoin block based on the proof of work and the header chain, and
since the cost of producing a valid block is expensive, we argue that \sys should be able to obtain valid block from the Bitcoin network.
Also, 
\sys maintains a Merkle Hash Tree (MHT) of the ORAM trees and therefore prevents malicious tampering by verifying all encrypted data fetched from the untrusted memory using the MHT.
All encrypted data fetched from the untrusted memory is verified using
the MHT.

%
\paragraph{Claim 5. The multiple threads involved do not create synchronization issues.}
Here, it is worth-noting that multiple threads are only involved while
accessing the Read Tree of \sys.
Thanks to the optimized read operation, \sys does not run into synchronization
bugs since there is no memory region that could be simultaneously written to
by more than one thread.
In particular, each thread shares the position map but only reads from the
position map.
Each thread contains its own stash memory which is written to separately by
each thread.

\paragraph{Claim 6. The memory interactions within the enclave are
side-channel-resistant.}
The design of $T^3$ incorporates defenses against the side-channel
threats~\cite{Xu15ControlledChannel,hid-sgx-sidechannel-usenix17,shadow-branch-lee-usenix17}
plaguing Intel SGX.
In particular, we used ORAM operations to hide all data access patterns on the untrusted memory region, 
and we incorporated similar oblivious operation techniques introduced in~\cite{racoon,ndss-AhmadKSL18,
SasyGF18-zero-trace} to prevent operations inside the enclave from leaking sensitive information.
Finally, the implementation of \sys is also secure against branch-prediction attacks since each
individual operation (e.g., accessing Read Tree, updating Write Tree etc.) takes the same sequence
of branches and therefore reveals no information to the attacker, from the accessed branches.

\subsection{Denial of Service Attacks from Malicious Clients}

While the design of \sys is pratical, a malicious client can still incur a large processing time on the server by creating lots of addresses and sending large number of requests for those requests. 
One way to mitigate such attack is to apply fees on users of the service.
Another approach to mitigate denial of service attack is to use a cuckoo filter~\cite{cuckoo-filter-Fan} to load all addresses from the UTXO set. Upon receiving requests from client, the managing enclave can verify if the address matches the filter as well as the proof of ownership of that address before performing ORAM accesses. Moreover, since Cuckoo filter data structure supports deletion operation, the system can add and remove addresses when performs updating. 
In other word, in order to perform the denial of service attack, clients need both the proof of ownership as well as a certain amount of Bitcoin in each address. 
Hence, it will cost more for the client to perform such attack.

\subsection{Other Goals Achieved by \sys}
In this subsection, in addition to the \textbf{Privacy} goal describe in \cref{sub:privacy}, we explain how \sys achieves the other goals mentioned in
\autoref{sub:goal}.

\paragraph{Validity.} Under the assumption that the adversary does not have enough
computational power to form a new Bitcoin block, the system will only obtain valid 
transaction by verifying the Merkle root and the proof of work of the Bitcoin block.

\paragraph{Completeness.} By  offering  different  ways  of  mapping  between  Bitcoin  addresses  and  ORAM  block  id,  we  can  offer services to $92-96\%$ of all clients
with some trade-off between storage overhead and performance.

\paragraph{Efficiency.} Our  contribution  to  efficiency  is  threefold.  First, our 
system  is  able  to  handle  bursty  requests  from  client concurrently. The core
idea is to separate the effect of a standard ORAM access into different enclaves.
Thus, the multiple reading enclave  can  concurrently  perform  read  operations at
the  same  time  that  the writing enclave  can  perform  a  non-blocking Evict 
procedure on the other tree. Second, by having two ORAM trees, we minimize the
downtime of the system by having the writing enclave performed updates on one
tree and reading enclave handled clients’ requests on the other tree. 
The server  downtime  depends  on  the  number  of  requests  that  the system 
receives  when  the writing enclave  performs  ORAM updates  on  the original ORAM 
tree.  Finally,  by  enforcing clients  to  provide  the  proof  of  ownership  of
the  address  and assuming  that  a  honest  client  is  rational,  we prevent  other clients from querying addresses that do not belong to them.

\section{Related Work}
\label{sec:related-work}

\paragraph{General SGX Systems.}
Haven~\cite{haven} is a pioneering work on SGX computing
enabling native application SGX porting on windows.
Graphene~\cite{graphene} provides a linux-based LibOS for
SGX programs.
Ryoan~\cite{ryoan} retrofits Native Client to provide sandboxing
mechanisms for Intel SGX.
Eleos~\cite{eleos} provides a user-space extension of enclave
memory using custom encryption.
\sys uses some concepts from Eleos especially in the way we store
the ORAM tree using custom encryption outside the SGX enclave.

\paragraph{SGX Side-channels.}
There are three main memory-based side-channel vulnerabilities disclosed
within Intel SGX, namely, page table-based attacks~\cite{Xu15ControlledChannel},
cache-based attacks~\cite{cache-based-attack}, and branch-prediction
attacks~\cite{hid-sgx-sidechannel-usenix17}.
Furthermore, since SGX relies on the untrusted OS for system-call
handling, it is also vulnerable to IAGO attacks~\cite{iago-attack}.
Leaky Cauldron~\cite{leaky-cauldron} presents an overview of the possible attack vectors
against SGX programs.
\sys is secure against all disclosed memory-based side-channels since it
uses oblivious RAM (ORAM) to protect the access-patterns.
Furthermore, \sys uses oblivious memory primitives to secure the runtime
ORAM operations as well as its library.

\paragraph{Oblivious Systems.}
Raccoon~\cite{racoon} provided a technique to protect a small part of
a user program against all digital side-channels.
\textsc{Obliviate}~\cite{ndss-AhmadKSL18} and \textsc{ZeroTrace}~\cite{SasyGF18-zero-trace}
used ORAM-based operations to protect files and data arrays respectively
inside Intel SGX. 
Thang Hoang et al.~\cite{thang-hoang:posup-popets} proposed a combination of TEE and ORAM to design oblivious search and update platform for large dataset.
Eskandarian et a.~\cite{oblidb} leveraged Intel SGX and Path ORAM to propose oblivious SQL database management system.

Recently, Chakraborti et al. proposed a new parallel ORAM scheme called ConcurORAM~\cite{concurORAM-chakraborti}. 
Similar to the \sys design, ConcurORAM also uses two-tree structure to propose a non-blocking eviction procedure, and the system periodically synchronizes two trees to maintain the privacy of the user's access pattern. 
In ConcurORAM, the scheme requires the client to download the query log and the result log to learn about ongoing queries before requesting ORAM accesses.
Hence, if we combine ConcurORAM along with Intel SGX to design this system, the use of query and result logs introduces additional storage overhead to the limited storage capacity of the Intel SGX.
More importantly, the author also noted that ConcurORAM cannot be trivially extended to a recursive ORAM construction because of concurrent data structure accesses. 
However, if ConcurORAM can be implemented into a recursive ORAM construction, we believe that ConcurORAM can be an interesting alternate solution for the ORAM scheme used in the design of \sys.  

Another interesting parallel ORAM construction is TaoStore~\cite{taostore-sahin}. 
TaoStore assumes a trusted proxy that handles concurrent client's requests, and the proxy runs a scheduler to make sure that there are no conflicting queries while preventing no information leakage. 
However, similar to ConcurORAM, the implementation of TaoStore is limited to the non-recursive construction of Path ORAM which is not suitable when combining with TEE with limited trusted memory capacity. 
This work aims to design a simpler design that is suitable for any flavor of tree-based ORAM schemes.  


\paragraph{TEE for cryptocurrencies.} The research community has investigated different ways of combining TEE with blockchain to both improve privacy and scalability of blockchains. 
Obscuro~\cite{obscuro-muoi-tran} is a Bitcoin transaction mixer implemented in Intel SGX that addresses the linkability issue of Bitcoin transactions.
Teechan~\cite{teechan} is an off-chain payment micropayment channel that harnesses TEE to increase transaction throughput of Bitcoin. Bentov et al. proposed a new design that uses Intel SGX to build a real-time cryptocurrency exchange. 
Another example is the Towncrier system~\cite{towncrier-Zhang} that uses TEE for securely transferring data to smart contract. 
Another prominent example is Ekiden~\cite{ekiden-Cheng} which proposed off-chain smart contract execution using TEE. Finally, ZLite~\cite{ZliTE} system is another example which used ORAM and TEE to provide SPV clients with oblivious access. However, similar to BITE, ZLite employed non-recursive \pathORAM as it is, and thus, the scalability and efficiency of the system is inherently limited due to the non-concurrent accesses.



Osuntokun et al.~\cite{osuntokun-client-filter} recently present a new proposal for Bitcoin SPV clients. 
This proposal is the building block for systems like Neutrino.
In particular, each block will have its own Bloom filter. 
The SPV client first fetches the filter from the full client and decides to download the block from another client if transactions of interested are in the block. 
This approach, however, introduces an additional communication overhead to the client. 
In particular, a client with lots of transactions scattered among different blocks needs to download lots of full blocks, and performing verification of block can be expensive for the resource-constrained client. 
This approach does not necessarily provide more privacy for the SPV client as the full client still learn the block that the addresses belong to.


\section{Conclusion} 
\label{sec:conclusion}

In this paper, we developed a system design that supports a large-scale oblivious search on unspent transaction outputs for Bitcoin SPV clients while efficiently maintains the state of the Bitcoin $\mathsf{UTXO}$ set via an oblivious update protocol. 
Our design leverages the TEE capabilities of Intel SGX to provide strong privacy and security guarantees to Bitcoin SPV client even with the presence of a potentially malicious server.
Moreover, by putting reasonable assumptions on the accessing frequency of the SPV clients, we present novel ORAM construction that offers both privacy and efficiency to the clients.
We showed that the prototype of the system is much more efficient than the use of standard ORAM construction as it is. 
In particular, due to the use of two ORAM trees in the design of \sys, we improve the performance of an ORAM access by two time and allow the system to handle concurrent client's requests.
Also, our implementation shows one order of magnitude performance gain when combining recursive ORAM construction the current existing construction to stress the importance of using recursive ORAM construction in TEE with restricted memory.
Finally, while the applicability of \SystemName in cryptocurrencies beyond Bitcoin is apparent, we believe our work will motivate further research on oblivious memory with the restricted access patterns.

\balance
{\small
\bibliographystyle{plain}
\bibliography{reference}
}
\end{document}